%% file: main.tex
\documentclass[11pt]{article}
\input{macros}

\begin{document}

\title{On Constructing Spanners from Random Gaussian Projections}

\author{
  Sepehr Assadi\footnote{(sepehr.assadi@rutgers.edu) Department of Computer Science, Rutgers University.}
  \and
  Michael Kapralov\footnote{(michael.kapralov@epfl.ch) School of Computer and Communication Sciences, EPFL.}
  \and
  Huacheng Yu\footnote{(yuhch123@gmail.com) Department of Computer Science, Princeton University.} 
  }

\maketitle

\pagenumbering{roman}

\begin{abstract}

\bigskip

Graph sketching is a powerful paradigm for analyzing graph structure via linear measurements introduced by Ahn, Guha, and McGregor (SODA'12) that has since found numerous applications in streaming, distributed computing, and massively parallel algorithms, among others. Graph sketching has proven to be quite successful for various problems such as connectivity, minimum spanning trees, edge or vertex connectivity, and cut or spectral sparsifiers. Yet, the problem 
of approximating shortest path metric of a graph, and specifically computing a spanner, is notably missing from the list of successes. This has turned the status of this fundamental problem into one of the most longstanding open questions in this area. 

\smallskip

We present a partial explanation of this lack of success by proving a {strong lower bound} for a \underline{large family of graph sketching algorithms} that encompasses prior work on spanners and many (but importantly not also all) related cut-based problems mentioned above.  
Our lower bound matches the algorithmic bounds of the recent result of Filtser, Kapralov, and Nouri (SODA'21), up to lower order terms, for constructing spanners via the same graph sketching family. This establishes near-optimality of these bounds, at least restricted to this family
of graph sketching techniques, and makes progress on a conjecture posed in this latter work.

\end{abstract}

\clearpage

\setcounter{tocdepth}{3}
\tableofcontents

\clearpage

\pagenumbering{arabic}
\setcounter{page}{1}

\input{intro.tex}

\input{prelim.tex}

\input{high-level.tex}

\input{kl.tex}

\input{eff_resist}
\input{expander_sample}

\section*{Acknowledgements}
Sepehr Assadi was supported in part by a
NSF CAREER Grant CCF-2047061, a Google Research gift, and a Fulcrum award from Rutgers Research Council. Michael Kapralov was supported by ERC Starting Grant 759471. Huacheng Yu was supported by a Simons Junior Faculty Award.


\newcommand{\etalchar}[1]{$^{#1}$}

\clearpage

\part*{Appendix}
\appendix
\input{implement}

\input{background}

\end{document}

%% file: macros.tex
\usepackage[utf8]{inputenc}
\usepackage[margin=1in]{geometry}
\usepackage{amssymb,amsmath,amsthm}
\usepackage{verbatim}
\usepackage{dsfont}
\usepackage{algorithm}
\usepackage[noend]{algpseudocode}
\usepackage{thmtools,thm-restate}

\usepackage{mdframed}
\usepackage{thmtools}
\usepackage{mathtools}
\usepackage{enumitem}
 
 \usepackage{bm,xspace}
 
\usepackage[usenames,dvipsnames]{xcolor}

\usepackage{tcolorbox}

\newenvironment{tbox}{\begin{tcolorbox}[
    enlarge top by=5pt,
    enlarge bottom by=5pt,
     breakable,
     boxsep=0pt,
                  left=4pt,
                  right=5pt,
                  top=10pt,
                  arc=0pt,
                  boxrule=1pt,toprule=1pt,
                  colback=white
                  ]
  }
{\end{tcolorbox}}

\usepackage{tikz}
\usetikzlibrary{arrows}
\usetikzlibrary{arrows.meta}
\usetikzlibrary{shapes}
\usetikzlibrary{backgrounds}
\usetikzlibrary{positioning}
\usetikzlibrary{decorations.markings}
\usetikzlibrary{patterns}
\usetikzlibrary{calc}
\usetikzlibrary{fit}
\usetikzlibrary{snakes}
\tikzset{vertex/.style={circle, black, fill=Yellow, line width=1pt, draw, minimum width=8pt, minimum height=8pt, inner sep=0pt}}
	
\tcbuselibrary{skins,breakable}
\tcbset{enhanced jigsaw}

\usepackage[normalem]{ulem}
\usepackage[compact]{titlesec}

\definecolor{DarkRed}{rgb}{0.5,0.1,0.1}
\definecolor{DarkBlue}{rgb}{0.1,0.1,0.5}

\usepackage{nameref}
\definecolor{ForestGreen}{rgb}{0.1333,0.5451,0.1333}
\definecolor{Red}{rgb}{0.9,0,0}
\usepackage[linktocpage=true,
  pagebackref=true,colorlinks,
  linkcolor=DarkRed,citecolor=ForestGreen,
  bookmarks,bookmarksopen,bookmarksnumbered]
  {hyperref}
\usepackage[noabbrev,nameinlink]{cleveref}
 
 \newcommand{\tvd}[2]{\ensuremath{\norm{#1 - #2}_{\mathrm{tvd}}}}
\newcommand{\Ot}{\ensuremath{\widetilde{O}}}
\newcommand{\eps}{\ensuremath{\varepsilon}}
\newcommand{\Paren}[1]{\Big(#1\Big)}
\newcommand{\Bracket}[1]{\Big[#1\Big]}
\newcommand{\bracket}[1]{\left[#1\right]}
\newcommand{\paren}[1]{\ensuremath{\left(#1\right)}\xspace}
\newcommand{\card}[1]{\left\vert{#1}\right\vert}

 \DeclareMathOperator*{\Exp}{\ensuremath{{\mathbb{E}}}}

 \newcommand{\kl}[2]{\ensuremath{\mathbb{D}(#1~||~#2)}}
\newcommand{\II}{\ensuremath{\mathbb{I}}}

\newcommand{\mireal}[1][]{
  \ifx\relax#1\relax%
    \II(\mione \,; \mitwo)%
  \else%
    \II(\mione \,; \mitwo\mid #1)%
  \fi
}

\newcommand{\R}{\mathbb{R}}

\newtheorem{theorem}{Theorem}
\newtheorem{lemma}{Lemma}[section]

\newtheorem{claim}[lemma]{Claim}
\newtheorem{proposition}[lemma]{Proposition}
\newtheorem{fact}[lemma]{Fact}

\newcommand{\norm}[1]{\left\lVert#1\right\rVert}

\newtheorem{mdresult}{Result}
\newenvironment{result}{\begin{mdframed}[backgroundcolor=lightgray!40,topline=false,rightline=false,leftline=false,bottomline=false,innertopmargin=2pt]\begin{mdresult}}{\end{mdresult}\end{mdframed}}

\newtheorem*{lemma*}{Lemma}
\newtheorem*{proposition*}{Proposition}
\newtheorem{definition}[lemma]{Definition}

\newcommand{\expect}{{\mathbb E}}
\newcommand{\var}{{\mbox{Var}}}

\newcommand{\tr}{\text{tr}}

\newcommand{\poly}{\text{poly}}
\newcommand{\vol}{\text{vol}}
\newcommand{\spn}{\mathrm{span}}

\newcommand{\dmin}{d_{\min}}

\newcommand{\bOne}{\mathbf{1}}
\newcommand{\bA}{\mathbf{A}}
\newcommand{\bD}{\mathbf{D}}
\newcommand{\bM}{\mathbf{M}}

\newcommand{\oM}{\overline{\mathbf{M}}}
\newcommand{\bI}{\mathbf{I}}
\newcommand{\bL}{\mathbf{L}}
\newcommand{\be}{\mathbf{e}}
\newcommand{\smp}{\mathrm{sp}}
\newcommand{\E}{\mathop{\mathbb{E}}}

\newcommand{\fir}{\textrm{first}}
\newcommand{\bU}{\mathbf{U}}
\newcommand{\bSigma}{\mathbf{\Sigma}}

\usepackage{todonotes}

\allowdisplaybreaks

\crefname{equation}{eq}{Eq}
\creflabelformat{equation}{(#1)#2#3}
  
{\makeatletter
  \gdef\xxxmark{%
    \expandafter\ifx\csname @mpargs\endcsname\relax 
    \expandafter\ifx\csname @captype\endcsname\relax 
    \marginpar{xxx}
    \else
    xxx 
    \fi
    \else
    xxx 
    \fi}
  \gdef\xxx{\@ifnextchar[\xxx@lab\xxx@nolab}
  \long\gdef\xxx@lab[#1]#2{{\bf [\xxxmark #2 ---{\sc #1}]}}
  \long\gdef\xxx@nolab#1{{\bf [\xxxmark #1]}}
  \long\gdef\xxx@lab[#1]#2{}\long\gdef\xxx@nolab#1{}%
}

\newcommand{\bB}{\mathbf{B}}
\newcommand{\bS}{\mathbf{S}}
\newcommand{\bg}{\mathbf{g}}

\newcommand{\ustar}{u^{*}}
\newcommand{\vstar}{v^{*}}
\newcommand{\estar}{e^{*}}

\newcommand{\bp}{\bm{p}}

\newcommand{\eff}[2]{R_{\textnormal{eff}}^{#1}(#2)}

\newcommand{\KLD}[3]{\mathbb{D}_{#3}(#1\,\,||\,\, #2)}

\newcommand{\dist}{\mathcal{D}}

\newcommand{\dists}{\dist^{\textnormal{\texttt{smpl}}}}

\newcommand{\distspan}{\mu}

\newcommand{\NN}{\mathcal{N}}

\newcommand{\set}[1]{\{#1\}}

\newcommand{\alg}{\ensuremath{\mathcal{A}}}

\newcommand{\bSv}{\mathcal{S}}
\newcommand{\bpv}{\mathcal{P}}

%% file: intro.tex
\section{Introduction}

Analyzing structure of different objects via random linear projections, also known as \emph{sketching}, is a fundamental paradigm that arise in various contexts. 
Canonical examples of this approach include dimensionality reduction results such as Johnson-Lindenstrauss lemma~\cite{JohnsonL84}, sparse recovery results in compressed sensing~\cite{Donoho06}, 
approximation algorithms for large matrices~\cite{Sarlos06,Woodruff14}, or various sketches for statistical estimation such as AMS sketch~\cite{AlonMS96}, count sketch~\cite{CharikarCF02}, or count-min sketch~\cite{CormodeM04} in data streams. 

A pioneering work of~\cite{AhnGM12a} initiated \emph{graph sketching} that considers this paradigm for graphs. 
A graph sketching algorithm samples a sketching matrix $A$ from a fixed distribution, independent of the input graph $G$, and compute $A \cdot R(G)$ where $R$ is a suitable representation of $G$ chosen by the algorithm designer, say, its adjacency matrix, Laplacian, or (signed) edge-incidence matrix. The  algorithm then uses $A \cdot R(G)$, referred to as the \emph{sketch}, to (approximately) discover properties of $G$ with no further access to $G$, e.g.,
to determine whether or not $G$ is connected. Assuming one can design a sketching matrix $A$ with ``few'' rows, this approach leads to sketches that can be stored and updated efficiently and be used to recover fundamental properties of $G$. 

The linearity of the sketches and the natural ``composability'' guarantee that comes with it makes graph sketching a versatile tool in many applications. For instance, graph sketching is the de facto method of algorithm design for \emph{dynamic streaming} algorithms that process a graph specified via a sequence of edge insertions and deletions; see, e.g.~\cite{AhnGM12a,AhnGM12b,KapralovLMMS14,AssadiKLY16,KapralovMMMNST20,FiltserKN21}. 
Graph sketching works seamlessly in this model as linearity of sketches allows one to update them easily after each update in the stream. It is even known that this method is {universal} 
for dynamic streams under certain (strong) assumptions on length of the stream~\cite{LiNW14,AiHLW16} (see also~\cite{KallaugherP20} for necessity of these assumptions). Another model 
that has benefited greatly from graph sketching is that of \emph{distributed sketching} (a.k.a. simultaneous communication model or broadcast congested clique) wherein every vertex is a processor that sees only edges incident on the vertex and its task 
is to communicate a small message, simultaneously with other vertices, that allows a referee to solve the problem; see, e.g.~\cite{BeckerMNRST11,AhnGM12a,AhnGM12b,BeckerMRT14,NelsonY19,AssadiKO20,Yu21}. Finally, graph sketching has also been a powerful tool for designing distributed or massively parallel algorithms; see, e.g.~\cite{AhnGM12a,AhnGM12b,HegemanPPSS15,GhaffariP16,PanduranganRS18,JurdzinskiN18,FernandezW020,FiltserKN21}. 

All these considerations have turned graph sketching into a highly attractive solution concept in the last decade since their introduction in~\cite{AhnGM12a}. We now have 
 efficient sketches, that often match existentially optimal bounds up to poly-log factors\footnote{For instance, sketches of size $O(n\log^3{n})$ for spanning forests of $n$-vertex graphs~\cite{AhnGM12a} compared to existential bound of $\Omega(n\log{n})$ bits  to store the spanning forest.},  for various fundamental problems such as connectivity~\cite{AhnGM12a}, minimum spanning trees~\cite{AhnGM12a}, edge connectivity~\cite{AhnGM12a}, vertex connectivity~\cite{GuhaMT15}, cut sparsifiers~\cite{AhnGM12b}, 
spectral sparsifiers~\cite{KapralovLMMS14,KapralovMMMNST20}, graph coloring~\cite{AssadiCK19}, densest subgraph~\cite{McGregorTVV15}, 
and others. 

\paragraph{Graph sketching for spanners.} We study graph sketching for the problem of computing \emph{spanners} that (approximately) preserve the {shortest path metric} of the input graph. 
Formally, 
\begin{definition}\label{def:spanner}
	A subgraph $H$ of a graph $G=(V, E)$ is a $d$-spanner of $G$ for some integer $d \geq 1$, called the \underline{stretch} of the spanner, if for every pair $u, v\in V$ one has 
\[
dist_G(u, v)\leq dist_H(u, v)\leq d \cdot dist_G(u, v),
\]
where $dist_*(\cdot,\cdot)$ stands for the shortest path metric of the corresponding graph. 
\end{definition}
For every integer $k \geq 1$, every $n$-vertex graph $G=(V,E)$ admits a $(2k-1)$-spanner with only $O(n^{1+1/k})$ edges which is also existentially optimal under the widely-believed {Erd\H{o}s Girth Conjecture}. 
For instance, every graph admits an $O(\log{n})$-spanner on $O(n)$ edges. 

Spanners are notably absent from the list of successes in graph sketching. Indeed, despite the significant attention given to sketching spanners, see, e.g.,~\cite{AhnGM12b,KapralovW14,FernandezW020,FiltserKN21,ElkinT21}, until very recently, it was not even known whether an $o(n)$-spanner can be recovered via sketches of $O(n)$ size. 
The work of~\cite{FiltserKN21} made the first progress on this 
problem in nearly a decade by presenting an $O(n^{2/3})$-spanner using sketches of $\Ot(n)$ size\footnote{We use $\Ot(\cdot)$ and $\widetilde{\Omega}(\cdot)$ notation to hide poly-log factors.}, or more generally a $d$-spanner using sketches of size $\Ot(n^2/d^{3/2})$.
But such bounds are still quite far from existential bounds on spanners dictated by the girth conjecture. Yet, no non-trivial lower bounds are known for this problem\footnote{This state-of-affairs is in sharp contrast with another widely-studied problem of finding 
large matchings which is also absent from the list of successes in graph sketching; for the matching problem, \emph{asymptotically tight} lower bounds which are much stronger than existential bounds are known; see~\cite{AssadiKLY16,DarkK20,AssadiS22}.}, 
beside the work of~\cite{NelsonY19} (see also~\cite{Yu21}) that proves that finding \emph{any} spanning tree requires sketches of size $\Omega(n\log^3{n})$ bits (namely, a lower bound for any spanner of finite stretch). 

The lack of progress on understanding graph sketching for spanners have also been consequential in other computing models that use graph sketching as their primary tool, most notably, the dynamic streaming model. Indeed, complexity of spanners has been a tantalizing open question in the dynamic streaming model already since its introduction in~\cite{AhnGM12a} (for insertion-only streams, 
optimal algorithms that essentially match existential bounds under Erd\H{o}s girth conjecture have been known since the introduction of the model itself in~\cite{FeigenbaumKMSZ04}; see, also~\cite{BaswanaS07,FeigenbaumKMSZ08}). 

This state-of-affairs raises the following question: \emph{What is the best stretch-vs-size tradeoff possible for constructing spanners via graph sketching?} We make progress on this longstanding open question 
by proving a nearly-tight lower bound for a large family of graph sketching algorithms that encompasses prior  work on spanners in~\cite{FiltserKN21} and most other closely related problems. 

\subsection{Our Contribution}\label{sec:results}

We prove \textbf{a lower bound on the size of a special case yet general family of sketches for graph spanners}. 
This family, that shall be defined shortly, contains the prior sketching algorithm of~\cite{FiltserKN21} for graph spanners -- our lower bound matches their bound up to lower order terms and is thus \textbf{nearly-optimal}. 
In addition, this family also contains many prior sketching algorithms for  ``cut-based'' problems such as connectivity~\cite{AhnGM12a}, vertex connectivity~\cite{GuhaMT15}, and
spectral sparsifiers~\cite{KapralovLMMS14} (and thus also cut sparsifiers). We now elaborate more on our results, starting with the definition of our sketches, which we call \emph{random Gaussian sketches}. 

\paragraph{Random Gaussian sketches.} To date, the main success of graph sketching has been for \emph{cut-based} problems~\cite{AhnGM12a,AhnGM12b,KapralovLMMS14,GuhaMT15}. These sketches all work by encoding a graph $G$ as its ${n \choose 2} \times n$ \emph{signed edge-incidence matrix} $B(G)$ (see~\Cref{sec:model}) and then apply a sketching matrix $A$ with few rows on the left to obtain the sketch $A \cdot B(G)$. The power of these sketches comes from surprisingly powerful cancellations that the use of the signed edge incidence matrix enables. In addition, the sketching matrix $A$ of in these approaches implements 
a \emph{sparse recovery} scheme on carefully chosen random subgraphs of the input graphs (e.g. uniformly random subgraphs of the input graph in the case of connectivity~\cite{AhnGM12a}, cut sparsifiers~\cite{AhnGM12b}, or spectral sparsifiers~\cite{KapralovLMMS14}, and sampled vertex induced subgraphs in the case of spanners~\cite{FiltserKN21}). 

To give a concrete example, let us consider the AGM sketches~\cite{AhnGM12a} for finding spanning forests. For any graph $G=(V,E)$ and any set of vertices $S \subseteq V$, adding up the columns of $B(G)$ corresponding to vertices in $S$, i.e., $\sum_{v \in S} B(G)^v$ gives 
us a vector with non-zero entries corresponding to edges of the cut $(S, V\setminus S)$. The linearity of matrix $A$ then allows us to obtain 
\[
A \cdot \paren{\sum_{v \in S}B(G)^v} = \sum_{v \in S} A \cdot B(G)^v,
\]
for a cut $S$ specified in the recovery phase. The sketching matrix $A$ itself is an $\ell_0$-sampler sketch that samples a non-zero entry of a vector $v$ given $A \cdot v$ (see~\cite{JowhariST11,KapralovNPWWY17}). 
An $\ell_0$-sampler sketch is typically implemented
via a simple sparse recovery sketch combined with a sampling matrix that samples the edges of the graph at $O(\log{n})$ geometrically decreasing rates. Combined with the above approach, we can thus sample an edge from any cut of the graph specified in 
the recovery phase. The algorithm of~\cite{AhnGM12a}  heavily builds on this subroutine by implementing Borůvka's algorithm for growing connected components via using these sketches to find an outgoing edge
from each component in each step. 

In this paper, we focus on this family of sketches where the sparse recovery scheme is implemented using {\em random Gaussian projections}. This means that each \emph{row}
of the sketching matrix is of the type $g \cdot S$ where $S$ is an ${{n}\choose{2}} \times {{n}\choose{2}}$-dimensional diagonal sampling matrix---where $S_{(u,v),(u,v)} = 1$ iff $(u,v)$ is sampled---and $g$ is an ${{n}\choose{2}}$-dimensional vector of independent Gaussian variables:
\[
	\begin{bmatrix}
		& & & & g & & & & \\
	\end{bmatrix}_{1 \times {{n}\choose{2}}} \times 
	\begin{bmatrix}
		& & & & & &  \\
		& & & & & &  \\
		& & & S  & &  \\
		& & & & & &  \\
		& & & & & & 
	\end{bmatrix}_{{{n}\choose{2}} \times {{n}\choose{2}}} \times
	\begin{bmatrix}
		&  &   \\
		&  &   \\
		& B(G) &   \\
		&  &   \\
		&  &  
	\end{bmatrix}_{{{n}\choose{2}} \times n} = 
	\begin{bmatrix}
		& g \cdot S \cdot B(G) &  \\
	\end{bmatrix}_{1 \times {n}}.
\]
The entire sketch is obtained by taking $s$ such rows where sampling matrices can be correlated but Gaussian vectors are independent. The recovery algorithm is given sampling matrices and the sketch but \emph{not} Gaussian variables. 
We refer to $s$ as the \emph{dimension} of the sketch (thus size of the sketch is $O(s \cdot n)$). See~\Cref{sec:model} for formal definitions. 

\emph{General ``power'' of random Gaussian sketches?} In~\Cref{app:implement}, we show this family of sketches can implement many (but importantly \emph{not} all) prior cut-based sketching algorithms in~\cite{AhnGM12a,KapralovLMMS14,GuhaMT15}, and most importantly the spanner sketch of~\cite{FiltserKN21}. But we also point out that these sketches are \emph{not} universal and one can easily construct problems where the power of these sketches does not match general sketching algorithms\footnote{Consider recovering the induced subgraph of the input on the first $\sqrt{n}$ vertices. A sparse recovery algorithm 
that spends $O(\sqrt{n})$ bits per each of these $\sqrt{n}$ vertices gives a sketch of size $O(n)$ for this problem. However, any random Gaussian sketch requires a dimension of $\Theta(\sqrt{n})$ that \emph{cannot} be amortized over all vertices, leading to a sketch of size $O(n^{3/2})$  instead.\label{fn:not-universal}}. Perhaps more importantly, we assume that the recovery algorithm of these sketches is \emph{oblivious} to the Gaussian vectors used in the sketching matrix which means that the recovery algorithm has a \emph{partial} knowledge of the sketching matrix. A particular shortcoming of this is that while these sketches handle the ``main'' source of cancelations enabled by edge-incidence matrix, they do \emph{not} handle a ``secondary'' source of cancelation: to obtain sketches of subgraphs of the input by generating the sketching matrix again at the recovery phase, apply it on some recovered
part of the input, and subtract it from the original sketch (this approach is used in the edge connectivity and cut sparsifier sketch of~\cite{AhnGM12b} -- although we  note that random Gaussian sketches can recover a cut
sparsifier by instead implementing the algorithm of~\cite{KapralovLMMS14}). We thus see the merit of study of this family as arguably the ``most natural'' candidate for finding spanners, given their past successes for closely related problems.  

\paragraph{Our result.} We prove a near-optimal lower bound on the dimension of random Gaussian sketches for constructing spanners, or even returning the distance of two fixed vertices (see also~\Cref{thm:main}). 
\begin{result}\label{res:main}
	Any random Gaussian sketch for constructing a $d$-spanner with constant probability of success requires dimension $\Omega(n^{1-o(1)}/d^{3/2})$, or put differently, any random Gaussian sketch of dimension $s$ 
	can only achieve a stretch of $\Omega((n/s)^{2/3-o(1)})$. The lower bound applies even to the problem of approximating the distance of two fixed vertices. 
\end{result}

Our lower bounds in~\Cref{res:main} matches algorithmic bounds of~\cite{FiltserKN21} up to the $n^{o(1)}$ term for computing spanners via graph sketching (whose sketches fit 
the framework of random Gaussian sketches) for all stretch $d$. This establishes the optimality of these bounds at least among this popular family of graph sketching algorithms. We note that~\cite{FiltserKN21} conjectured optimality of their algorithmic 
bounds among \emph{all} graph sketching techniques; our bounds in~\Cref{res:main} makes partial progress towards settling this conjecture. 

Before moving on, we note that  for the case when dimension $s=\poly\!\log\!{(n)}$, corresponding
to sketches of size $\Ot(n)$,~\Cref{res:main} implies a lower bound of $n^{2/3-o(1)}$ on the stretch; this should be contrasted with the $O(\log{n})$ bound of existential results on the stretch of spanners with $O(n)$ edges, suggesting that computing spanner 
is much harder using graph sketching (specifically via random Gaussian sketches) compared to existential bounds and arbitrary algorithms. Finally, the lower bound holds even for the algorithmically easier problem 
of simply estimating distance of two fixed vertices in the graph, as opposed to recovering the entire shortest path metric via a spanner. 

\paragraph{Our techniques.} We consider a  family of hard instances that form a random chain of cliques of size $(n/d)$ with diameter $d$, and a single edge $\estar$ that connects two vertices
at distance $\Theta(d)$  together (see~\Cref{fig:dist}). It is easy to see that such $\estar$ should belong to every $o(d)$-spanner of the graph and we prove that no random Gaussian sketch of 
``small'' dimension can recover $\estar$. The proof  is through analyzing how much a \emph{single} random Gaussian projection can reveal information about $\estar$,
or a bit more formally, the KL-divergence between the resulting sketches of two neighboring graphs that only differ on $\estar$. The rest follows by summing up this information across the $s$ projections.

To prove the bound for a single projection, we use properties of Gaussian variables and KL-divergence to bound the information revealed about the edge $\estar$ by the \emph{effective resistance} of the \emph{sampled} subgraph of the input after applying the sampling matrix. We prove that the distribution of our input, combined with a \emph{hierarchical expander decomposition} of all edges of sampling matrix, implies that the sampled subgraph of the input form a chain of \emph{expanders} (with proper lower bounds on both expansion and minimum degree). This step requires analyzing expansion of vertex-sampled subgraphs of an expander which can be of independent interest. Lastly, we bound the effective resistance of the edge $\estar$ in 
this chain of expanders by exhibiting a proper \emph{electrical flow} in the graph using properties of expanders. 

\paragraph{Related work.} In a recent independent work Chen, Khanna and Li~\cite{ChenKL} showed, similarly to our work, a lower bound matching the sketching dimension of~\cite{FiltserKN21} for linear sketches that can support {\em continuous} weight updates (as opposed to sketches that are only required to work for unweighted graphs).  Thus, from the perspective of the ultimate result, the lower bound of~\cite{ChenKL} is incomparable to ours. Their lower bound works for more general sketches than ours (although still not universal), but assumes that these sketches work in the continuous weight update model; our lower bound assumes a special sketch structure, but works in the mode standard setting of unweighted graphs. There is quite a bit of overlap in techniques: both papers use expander decompositions and prove that expanders are preserved under vertex sampling (but the actual proofs of the corresponding lemmas are different).

%% file: prelim.tex

\section{Preliminaries}\label{sec:prelim}

\paragraph{Notation.} 
We use $\NN(\mu,\sigma)$ to denote the Gaussian distribution with mean $\mu$ and variance $\sigma^2$. 
For any distributions $P$ and $Q$, $\KLD{P}{Q}{}$ denotes the \emph{KL-divergence} of $P$ from $Q$ and 
$\tvd{P}{Q}$ is the \emph{total variation distance} between $P$ and $Q$. See~\Cref{sec:info} for the complete definitions. 

For a graph $G=(V,E)$ on $n$ vertices, we use $d_1,\ldots,d_n$ to denote the degrees of vertices in $G$. 
For any sets of vertices $S,T \subseteq V$, $E(S,T)$ denotes the set of edges between $S$ and $T$ and $\vol_G(S) := \sum_{v \in S} d_v$ denotes the \emph{volume} of $S$ (we drop the subscript when clear). 
The \emph{conductance} of $G$ is defined as 
\[
	\varphi(G) := \min_{S\subseteq V}\frac{\left|E(S,V\setminus S)\right|}{\min\{\vol(S), \vol(V\setminus S)\}}. 
\]
We say that $G$ is a \emph{$\varphi$-expander} if its conductance is at least $\varphi$. 

For a graph $G$, $\bA$ is the \emph{adjacency matrix}, $\bD$ is the \emph{degree diagonal matrix}, $\bB$ is the \emph{signed edge-incidence matrix}, $\bL$ is the \emph{Laplacian matrix}, and $\widetilde{\bL}$ is the \emph{normalized Laplacian matrix}. 
The \emph{spectral gap} of $G$ is defined as the second smallest eigenvalue of $\widetilde{\bL}$ which is related to the conductance via Cheeger's inequality (\Cref{lem:conduct_spgap}). Finally, $\eff{G}{u,v}$ 
denotes the \emph{effective resistance} between $u, v$ when treating edges of $G$ as resistors with unit resistance. See~\Cref{sec:spectral} for  definitions. 

We also use the following (variant of) expander decomposition that bounds the minimum degree of resulting expanders. The proof is a simple modification of standard decompositions, e.g. in~\cite{VempalaV00,SaranurakW19},
and is provided in~\Cref{sec:spectral} for completeness.

\begin{proposition}\label{prop:expander-decomposition}
	Let $G=(V, E)$ be any  graph on $n$ vertices and $m$ edges, and $\eps \in (0,1/2)$ and $\dmin \geq 1$ be parameters. The vertices of $G$ can be partitioned into subgraphs $H_1,\ldots,H_k$ such that: 
	\begin{enumerate}[label=($\roman*)$]
		\item Each $H_i$ is an $\eps$-expander with minimum degree $\dmin$; 
		\item At most $8\eps \cdot m\log{n} + n \cdot \dmin$ edges $E_0$ of $G$ do not belong to any subgraph $\set{H_i}_{i \in [k]}$. 
	\end{enumerate}
\end{proposition}

%% file: high-level.tex

\section{Main Result}\label{sec:high-level}

We formalize~\Cref{res:main} in this section. We start by defining the sketching model, using random Gaussian projections, that we study. We then present our lower bound for constructing spanners (and in general preserving shortest
path metric) using these sketches. Finally, we give the proof outline of this result here and postpone the proof of its  main ingredients to the subsequent sections.  

\subsection{Random Gaussian Projections and Sketches}\label{sec:model}

For an $n$-vertex graph $G=(V,E)$, its \textbf{signed edge-incidence} matrix is an ${{n}\choose{2}} \times n$-dimensional matrix $\bB = \bB(G)$ defined as follows:
\begin{itemize}
	\item Each column corresponds to a vertex $v$  and each row corresponds to a pair of vertices $(u,w)$;
	\item The entry $\bB_{(u,w),v}$ is either $+1$ if $(u,w)$ is an edge in $G$ and $v=u$, $-1$ if $(u,w)$ is an edge in $G$ and $v=w$, and $0$ otherwise. 
\end{itemize}
Note that for any edge $e=(u,v)$ of $G$, the corresponding row $(u,v)$ in $\bB$ has exactly one $+1$ at column $u$, one $-1$ at column $v$, and is otherwise $0$. A row $(u,v)$ of $\bB$ which does not have a corresponding edge in $G$ is the all-$0$ row. 

Our sketches are based on taking random Gaussian projections of matrix $\bB$, which roughly speaking correspond to sampling edge of $G$ (using any sampling scheme oblivious to the graph), and multiply a Gaussian vector 
with signed edge-incidence matrix of the resulting graph. Formally,  

\begin{definition}\label{def:g-projection}
	Let $G=(V,E)$ be an $n$-vertex graph and consider the following: 
	\begin{enumerate}[label=$(\roman*)$]
		\item \textbf{\emph{Sampling matrix}}: Let $\bS$ be a ${{n}\choose{2}} \times {{n}\choose{2}}$-dimensional \emph{diagonal} matrix with $0$-$1$-values on the diagonal. 
		Notice that the matrix $\bS \cdot \bB(G)$ is the edge-incidence matrix of the subgraph of $G$ obtained by picking only those edges of $G$ that their corresponding (diagonal) value in $\bS$ is $1$. 
		\item \textbf{\emph{Gaussian projection}}: Let $\bg$ be a ${{n}\choose{2}}$-dimensional vector of Gaussian random variables, where each entry is sampled independently from $\NN(0,1)$.
	\end{enumerate}
	A \textbf{\emph{random Gaussian projection}} of $G$ with respect to $\bS$ is an $n$-dimensional vector obtained by sampling $\bg \sim \NN(0,1)^{{n}\choose{2}}$, 
	and returning $\bp := \bg \cdot \bS \cdot \bB(G)$. 
\end{definition}

Using~\Cref{def:g-projection}, we can define the sketches we focus on as follows. 

\begin{definition}\label{def:g-sketches}
	Let $\Pi$ be a problem defined on $n$-vertex graphs $G=(V,E)$. A \textbf{random Gaussian sketch} for $\Pi$ is defined via the following pair: 
	\begin{enumerate}[label=$(\roman*)$]
		\item \textbf{\emph{Sketching matrices}}: A distribution $\dists$  on $s$-tuples of sampling matrices for some $s \geq 1$. 
		\item \textbf{\emph{Recovery algorithm}}: An algorithm that given $s$ sampling matrices $\bSv=(\bS_1,\ldots,\bS_s) \sim \dists$ and $s$ random Gaussian projection $\bpv=(\bp_1,\ldots,\bp_s)$  of any graph $G$ with respect to these sampling 
		matrices, returns a solution to $\Pi(G)$. 
	\end{enumerate}
	We refer to $s$ as the \underline{dimension} of the sketch (note that a sketch of dimension $s$ has size $O(s \cdot n)$). 
	
	\smallskip
	A random Gaussian sketch for a graph $G$ then consists of sampling the sketching matrices $\bSv$ from $\dists$ (independent of $G$), receiving random Gaussian projections $\bpv$, 
	and running the recovery algorithm on $(\bSv,\bpv)$  to return the solution. 
\end{definition} 

We emphasize that in~\Cref{def:g-sketches}, the recovery algorithm is given the sketching matrices used for random Gaussian projections explicitly, but is \emph{not} given  the Gaussian vectors themselves.

We note that our formalization of random Gaussian sketches is new to this paper, albeit it has been used implicitly in prior algorithmic results for in graph sketching literature. In~\Cref{app:implement}, we 
elaborate more on this connection and point out that how these sketches can be used to solve many of the canonical problems in graph sketching literature such as connectivity, minimum spanning tree, cut or spectral sparsifiers, and most closely related to ours, spanners. But we also emphasize that these sketches are \emph{not} universal -- see the discussion on the power of these sketches in~\Cref{sec:results}. 

\subsection{The Lower Bound}

The following is the formalization of~\Cref{res:main} that we prove. 

\begin{theorem}\label{thm:main}
	For any absolute constant $\delta \in (0,1)$, and integers $n \geq 1$ and $1 \leq d \leq n^{2/3-\delta}$, any random Gaussian sketch (\Cref{def:g-sketches}) that outputs a $d$-spanner 
	of every given $n$-vertex graph $G$ with probability at least $2/3$ has dimension (i.e., number of rows) 
	\[
	\Omega(\dfrac{n^{1-\delta}}{d^{3/2}}).
	\]
	Moreover, the lower bound continues to hold even if the algorithm is only required to answer the shortest path \emph{distance} 
	between two \emph{prespecified} vertices up to a factor of $d$.  
\end{theorem}
\Cref{thm:main} can alternatively be seen as proving that any random Gaussian sketch of dimension $s$ can only achieve a stretch of 
\[	
	\Omega((\dfrac{n}{s})^{2/3-\delta}),
\]
for any constant $\delta > 0$. In light of the result of~\cite{FiltserKN21}, the bounds obtained in~\Cref{thm:main} are optimal, up to $n^{o(1)}$-factors, for the entire range of dimension $s$ or stretch $d$.  
In particular,~\Cref{thm:main} implies that to obtain a $n^{2/3-\Omega(1)}$-spanner, one needs random Gaussian sketches of dimension $n^{\Omega(1)}$. This makes progress on a conjecture of~\cite{FiltserKN21} that stated 
the same bounds for \emph{arbitrary} sketches. 

Finally, we also mention that~\Cref{thm:main} works even for the problem wherein we are given two vertices $a$ and $b$ of the graph, and our goal is to simply determine the distance of $a$ and $b$ 
in the graph using the sketches. This problem is algorithmically easier than finding a spanner of the graph in that firstly, we do not need to pick subset of edges of the graph $G$ and can preserve the shortest path metric in any desired way, 
and secondly that we only need to maintain the distance between two vertices and not all pairs. Yet, effectively the entirety of our effort is to prove the result for spanners already and we get this stronger lower bound almost for free using standard ideas. 

In the rest of this section, we first present a hard input distribution used to establish~\Cref{thm:main}. We then state our main technical lemma that bounds the information revealed by a single random Gaussian projection 
on the graphs sampled from this distribution and show how this lemma easily implies the theorem. The next subsection then includes the proof outline of this technical lemma, whose main ingredients are postponed to the next sections. 

\subsection{A Hard Input Distribution}\label{sec:dist} 

For any sufficiently large $n,d > 0$, we define a hard  distribution $\distspan = \distspan(n,d)$ over $n$-vertex graphs. For simplicity, we prove the lower bound for $(d/2)$-spanners instead -- re-parameterizing $d$ then implies the same asymptotic lower bound for exact $d$-spanners as well (see~\Cref{fig:dist}). 

\begin{tbox}
\textbf{Distribution $\distspan(n,d)$.} A hard input distribution for $(d/2)$-spanners of $n$-vertex graphs. 
\begin{enumerate}[label=\emph{\arabic*.}, leftmargin=15pt]
	\item Partition the vertices $V$ into $d$ groups $V_1,\ldots,V_{d}$: each $v \in V$ is sent to one of the groups chosen uniformly at random. 
	\item Let $G$ be a graph obtained by placing a clique on each $V_{i} \cup V_{i+1}$ for $i \in [d-1]$. 
	\item  Sample a pair of vertices $(\ustar,\vstar) \in {{V}\choose{2}}$ independently and return the graph $G+ \estar$ for $\estar=(\ustar,\vstar)$. 
\end{enumerate}
\end{tbox}

\begin{figure}[h!]
	\centering
	\input{fig-dist}
	\caption{An illustration of $\distspan=\distspan(n,d)$ for $n=24$ and $d=8$. Any $4$-spanner of $G$ contains $\estar$.}\label{fig:dist}
\end{figure}
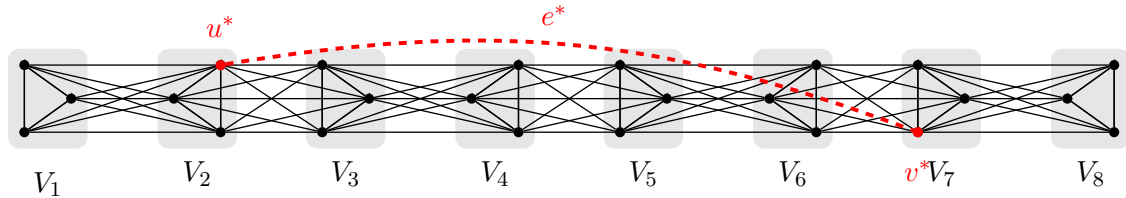

In the following, we use $\bB=\bB(G)$ to the denote the signed edge-incidence matrix of $G$; we also use $\bB(\estar)$ 
as the edge-incidence matrix of the $n$-vertex graph consisting of the single edge $\estar=(\ustar,\vstar)$. We emphasize that the final
graph output by the distribution is $G+\estar$ (this notation will be make the latter parts of the proof cleaner).

We first establish a straightforward property of graphs sampled from $\distspan$ in context of spanners. 

\begin{lemma}\label{lem:special-edge}
	With constant probability over the choice of $(G,\estar) \sim \distspan(n,d)$, every $(d/2)$-spanner of $G+\estar$ contains the edge $\estar=(\ustar,\vstar)$. 
\end{lemma}
\begin{proof}
	For a graph $G$ and pairs $(\ustar,\vstar)$ sampled from $\distspan(n,d)$, 
	\begin{align*}
		\Pr_{\ustar,\vstar}\paren{dist_G(\ustar,\vstar) > d/2} &= \frac{1}{d^2} \cdot \paren{O(d)+\sum_{i=1}^{d} \card{d/2-i}} > \frac{1}{5}, 
	\end{align*}
	where the equality holds since when $\ustar \in V_1$, $\vstar$ can be in $V_{d/2+2},\ldots,V_d$, when $\ustar \in V_2$, 
	$\vstar$ can be in $V_{d/2+3},\ldots,V_d$, and so on ($O(d)$ handles the  differences of even or odd choices of $d$ and $d/2$).

	Moreover, whenever the distance of $\ustar,\vstar$ in $G$ is more than $d/2$, any $(d/2)$-spanner of $G+\estar$ should contain the edge $\estar$, as otherwise 
	the distance between $\ustar$ and $\vstar$ in the spanner will be more than $d/2$ times their distance in $G+\estar$, violating the bound on the stretch of the spanner. 
\end{proof}

The following lemma is the main technical contribution of our work. Roughly speaking, this lemma bounds the ``information'' that can be learned about the edge $\estar$ in $\distspan$ using a \emph{single} sub-sampled Gaussian 
projection of a graph sampled from $\distspan$. 

\begin{lemma}\label{lem:main}
	Let $\bS$ be any sampling matrix and consider a single random Gaussian projection with respect to $\bS$. For $(G,\estar)$ sampled from $\distspan = \distspan(n,d)$, 
	\[
		\Exp_{G,\estar}\Bracket{\min\set{1,\KLD{\bg \cdot \bS \cdot \bB(G)}{\bg \cdot \bS \cdot (\bB(G)+\bB(\estar))}{\bg}}} = O(\frac{d^{3/2}}{n^{1-\delta}}),
	\]
	for any constant $\delta > 0$, where the KL-divergence is taken only over the Gaussian variables. 
\end{lemma}

Before getting to the proof of~\Cref{lem:main}, we show that it implies~\Cref{thm:main} immediately. 

\begin{proof}[Proof of~\Cref{thm:main} (assuming~\Cref{lem:main})]
	Let $(\dists,\alg)$ be any sub-sampled Gaussian sketch of dimension $s \geq 1$ for recovering a $(d/2)$-spanner. Consider a distribution $\distspan'$ on $n$-vertex graphs defined as follows: 
	\begin{itemize}
		\item Distribution $\distspan'$: Sample $(G,\estar)$ from $\distspan$ and $\theta \in \set{0,1}$ uniformly at random; if $\theta=0$, return $G$, otherwise return $G+\estar$. 
	\end{itemize}
	Let $G'$ be a graph sampled from $\distspan'$. 
	Suppose we sample $\bSv = (\bS_1,\ldots,\bS_s)$ from $\dists$ and receive sub-sampled Gaussian projections $\bpv = (\bp_1,\ldots,\bp_s)$ where for every $i \in [s]$, 
	$\bp_i = \bg_i \cdot \bS_i \cdot \bB(G')$ for a Gaussian vector $\bg_i$. Additionally, suppose we are even given $(\bSv,G,\estar)$, and thus the only unknown information is whether or not $\estar \in G'$ also, i.e., whether $\theta=1$ or not.  
	This way, we can run $\alg$, using $\bSv,\bpv$ as input, to obtain a $(d/2)$-spanner of $G'$: if $\estar$ belongs to this spanner, we declare $\estar$ is in $G'$ and otherwise we say it is not. 
	By~\Cref{lem:special-edge}, we are going to be able to determine the value of $\theta$ with probability $1/2+\Theta(1)$. This implies that over the distribution $\distspan'$, 
	\begin{align}
		\tvd{\bracket{(\bSv,G,\estar,\bpv) \mid \theta=0}}{\bracket{(\bSv,G,\estar,\bpv) \mid \theta=1}} = \Omega(1), \label{eq:LHS-sample}
	\end{align}
	as otherwise, by~\Cref{fact:tvd-sample}, we cannot estimate the value of $\theta$ with probability better than $1/2+o(1)$ given our input $(\bSv,G,\estar,\bpv)$ which is sampled from either $\mu' \mid \theta=0$ or $\mu' \mid \theta=1$. 
	We now have, 
	\begin{align*}
		\textnormal{LHS of~\Cref{eq:LHS-sample}} &\leq \Exp_{(\bSv,G,\estar)}  \tvd{\bracket{\bpv \mid \theta=0, \bSv,G,\estar}}{\bracket{\bpv \mid \theta=1, \bSv,G,\estar}} 
		\tag{by~\Cref{fact:tvd-chain-rule}, as the distribution of $(\bSv,G,\estar)$ is the same under both $\theta=0$ and $\theta=1$} \\
		&\leq \sqrt{\Exp_{(\bSv,G,\estar)} \min\set{1,\kl{{\bpv \mid \theta=0, \bSv,G,\estar}}{{\bpv \mid \theta=1, \bSv,G,\estar}}}} \tag{by Pinsker's inequality (\Cref{fact:pinsker}), the fact that TVD is bounded
		by $1$, and concavity of $\sqrt{\cdot}$} \\
		&= \sqrt{\Exp_{(\bSv,G,\estar)} \bracket{\min\set{1,\sum_{i=1}^{s}\kl{{\bp_i \mid \theta=0, \bSv,G,\estar}}{{\bp_i \mid \theta=1, \bSv,G,\estar}}}}} \tag{by chain rule of KL-divergence (\Cref{fact:kl-chain-rule}) as 
		$\bp_i$'s are now only function of $\bg_i$'s and so are independent} \\
		&\leq \sqrt{\sum_{i=1}^{s}\Exp_{(\bSv,G,\estar)} \Bracket{\min\set{1,\kl{{\bp_i \mid \theta=0, \bSv,G,\estar}}{{\bp_i \mid \theta=1, \bSv,G,\estar}}}}} \tag{we can take $\min$ inside 
		the summation to get an upper bound} \\
		&= \sqrt{\sum_{i=1}^{s} \Exp_{(\bS_i,G,\estar)} \Bracket{\min\set{1,\kl{{\bp_i \mid \theta=0, \bS_i,G,\estar}}{{\bp_i \mid \theta=1, \bS_i,G,\estar}}}}} \tag{as $\bp_i$ is only a function of $\bS_i$ conditioned on $G,\estar$} \\
		&=  \sqrt{\sum_{i=1}^{s} \Exp_{(\bS_i,G,\estar) \sim \distspan}\min\set{1,\KLD{\bg_i \cdot \bS_i \cdot \bB(G)}{\bg_i \cdot \bS_i \cdot (\bB(G)+\bB(\estar))}{\bg_i}}},
	\end{align*}
	where the last equality is because input graph $G'$ in $\distspan'$ is $G$ when $\theta=0$ and $G+\estar$ when $\theta=1$, and distribution of $(\bS_i,G,\estar,\bg_i)$ is the same under $\distspan$ and $\distspan'$. 
	
	Now given that $\bS_i \perp (G,\estar)$ in $\distspan$, each term in the RHS above is the same quantity upper bounded in~\Cref{lem:main}. Thus, combining~\Cref{eq:LHS-sample}, the above equation, and~\Cref{lem:main}, we get that 
	\[
			\Omega(1) = \tvd{\bracket{(\bSv,G,\estar,\bpv) \mid \theta=0}}{\bracket{(\bSv,G,\estar,\bpv) \mid \theta=1}} \leq \sqrt{s \cdot O(\frac{d^{3/2}}{n^{1-\delta}})},
	\]
	which implies that $s = \Omega(n^{1-\delta}/d^{3/2})$ as desired. This implies the first part of~\Cref{thm:main}. 
	
	The proof of the second part follows almost immediately from the above argument as follows. Consider the following distribution: 
	\begin{itemize}
		\item Distribution $\distspan''$: Sample $(G',\estar,\theta)$ from $\distspan'$. Add two new vertices $a$ and $b$ to the graph and add edges $(a,\ustar)$ and $(\vstar,b)$ to the graph as well. 
	\end{itemize}
	Let $G''$ be a graph sampled from $\distspan''$. Consider the distance between $a$ and $b$ in $G''$: if $\theta=1$ in the sampled $G'$, distance of $a$ and $b$ is $3$, otherwise, if $\theta=0$, 
	by the same argument as~\Cref{lem:special-edge}, the distance between $a$ and $b$ is more than $(d/2)$ with constant probability. This means that if our algorithm could simply estimate the distance of $a$ and $b$ 
	to within a factor of $(d/6)$, it can determine the value of $\theta$ with probability $1/2+\Theta(1)$. 
	
	Now if we further give $\ustar$, $\vstar$, and the Gaussian variables on all edges incident to $a$ or $b$ to the recovery algorithm, what the algorithm knows becomes the sketches of $G''\setminus \{a,b\}$ (by simply subtracting the corresponding Gaussians).
	Since the Gaussians revealed are independent of the sketch of $G''\setminus \{a,b\}$, the same exact argument as the first part now implies that the same lower bound of $s=\Omega(d^{3/2}/n^{1-\delta})$ on the sketch dimension. Given that the number of vertices in $G''$ is $n+2$, and by re-parameterizing $d$ with a constant factor, we obtain the desired lower bound. This concludes the proof of~\Cref{thm:main}. 
\end{proof}

\subsection{Proof Outline of~\Cref{lem:main}} 

We now present the proof outline of~\Cref{lem:main}, postponing the proof of its two main ingredients to the next two sections. 
For convenience, we restate~\Cref{lem:main} below. 

\begin{lemma*}[Restatement of~\Cref{lem:main}]
	Let $\bS$ be any sampling matrix and consider a single random Gaussian projection with respect to $\bS$. For $(G,\estar)$ sampled from $\distspan = \distspan(n,d)$, 
	\[
		\Exp_{G,\estar}\Bracket{\min\set{1,\KLD{\bg \cdot \bS \cdot \bB(G)}{\bg \cdot \bS \cdot (\bB(G)+\bB(\estar))}{\bg}}} = O(\frac{d^{3/2}}{n^{1-\delta}}),
	\]
	for any constant $\delta > 0$, where the KL-divergence is taken only over the Gaussian variables. 
\end{lemma*}
\noindent
To continue, we define the following notation for the sampling matrix $\bS$: 
\begin{itemize}
	\item $graph(\bS)$: the graph on $V$ containing all edges $e \in {{V}\choose{2}}$ where $\bS_{e,e} = 1$. 
	\item $m(\bS)$: the number of \emph{edges} in $graph(\bS)$. 
	\item $G(\bS)$: the subgraph of $G$ on edges that belong to $graph(\bS)$, i.e., $G(\bS) := G \cap graph(\bS)$. Note that this way we have, $\bB(G(\bS)) = \bS \cdot \bB(G)$ and $\bB(G(\bS)+\estar) = \bS \cdot (\bB(G)+\bB(\estar))$. 
\end{itemize}

\paragraph{Ingredient one: from KL-divergence to effective resistances.} The first key step of the proof of~\Cref{lem:main} is to relate
the KL-divergence term of~\Cref{lem:main} to \emph{effective resistance} of the edge $\estar$ in the underlying sampled graph. Formally, 

\begin{lemma}\label{lem:kl-eff-needed}
	For any sampling matrix $\bS$, any fixed $G$, and any pair of vertices $e=(u,v) \in {{V}\choose{2}}$, 
	\[
		\min\set{1,\KLD{\bg \cdot \bS \cdot \bB(G)}{\bg \cdot \bS \cdot (\bB(G)+\bB(e))}{\bg}} \leq 2 \cdot {\eff{G(\bS)+e}{u,v}}.
	\] 
\end{lemma}
We will apply~\Cref{lem:kl-eff-needed} to the choice of edge $\estar=e$ whenever $\estar$ belongs to $graph(\bS)$, i.e., when $\estar$ is sampled by the sampling matrix $\bS$. 
To prove~\Cref{lem:kl-eff-needed}, we first calculate the KL-divergence between two high-dimensional Gaussians in terms of their covariance matrices.
Then we observe that the covariance matrix of $\bg \cdot \bS \cdot \bB(G)$ is simply the \emph{Laplacian matrix} of $G(\bS)$.
The lemma is proved by plugging in the Laplacian matrices of $G(\bS)$ and $G(\bS)+e$, and applying the connection between effective resistance and Laplacian matrix.
The proof is provided in~\Cref{sec:kl-bounds}.

\paragraph{Ingredient two: bounding effective resistances via expanders.} Our strategy is now to bound the effective resistance of 
the edge $\estar$ in $G(\bS)+\estar$. To do so, we will identify a ``good''-\emph{expander} subgraph $H$ of the $graph(\bS)$ that contains the edge $\estar$, and then primarily focus 
on the edges of $H$ that appear in $G(\bS)$ to bound the effective resistance of $\estar$ also. The following lemma is the heart of the proof.

\begin{lemma}\label{lem:effective-large}
	For any sampling matrix $\bS$, suppose $H$ is any subgraph of $graph(\bS)$ which is an $\eps$-expander with min-degree $D$ for some $\eps > 0$ and $D \geq(\eps^{-8} \cdot n^{\delta}) \cdot d$ for a constant $\delta > 0$.
	For any edge $e=(u,v) \in H$,   
	\[
	\Exp_G\Bracket{\eff{G(\bS)+e}{u,v}} = O(\eps^{-4}) \cdot \Paren{\frac{d}{D} + \frac{d^3}{D^2}}.
	\]
\end{lemma}
We will use a hierarchical expander decomposition of $graph(\bS)$ to identify an expander that contains the edge $\estar$ and then apply~\Cref{lem:effective-large} to this expander and edge $\estar$. 
To prove~\Cref{lem:effective-large}, we first observe that adding edges to a graph could only decrease the effective resistance, and thus, it suffices to study the effective resistance of $(u, v)$ in $G\cap H$.
Note that $G\cap H$ randomly partitions the vertices into $d$ sets, and only keeps the edges of $H$ with endpoints in the same set or adjacent sets.
We then show that because $H$ is an expander with large min-degree, $G\cap H$ restricted to any two adjacent sets must also be an expander with large min-degree with high probability.
Hence, $G\cap H$ looks like a ``chain of expanders'' (which we call \emph{a balanced path of expanders}), where adjacent expanders have a constant fraction overlap.
Finally, we show that since $G\cap H$ overall is well-connected, if we place a unit electric flow from $u$ to $v$, the flow will be well-spread across the graph.
Most edges have a small current, i.e., a low potential difference.
Therefore, it allows us to argue that the potential difference between $u$ and $v$ is also small, i.e., the effective resistance between $u$ and $v$ is small.
The detailed proof is provided in~\Cref{sec:eff_upper}. 

\paragraph{Putting everything together.} We now put these two ingredients together to prove~\Cref{lem:main}. In order to be able to apply our second tool in~\Cref{lem:effective-large}, we need a 
\emph{hierarchical} expander decomposition of $graph(\bS)$, which shows that the edge $\estar$ is ``more likely'' to land in ``better'' expanders of $graph(\bS)$ for the purpose of~\Cref{lem:effective-large} -- here, ``better'' means an 
expander with a higher minimum degree (the parameter in~\Cref{lem:effective-large} that governs the final bound).


\begin{lemma}\label{lem:hierarchical}
	For every $t \geq 1$, we can partition edges of $graph(\bS)$ into $t$ sets $E_1(\bS),\ldots,E_t(\bS)$ such that: 
	\begin{enumerate}[label=$(\roman*)$]
		\item For any $i \leq t$, define $m_i(\bS) := \card{E_i(\bS)}$; then, $m_1(\bS) \leq m(\bS)$ and $m_{i+1}(\bS) \leq m_{i}(\bS)/{2}$.
		\item For any $i < t$, there is some $k_i \geq 1$ such that edges in $E_i(\bS)$ can be partitioned into $\eps$-expanders $H^i_1,\ldots,H^i_{k_i}$ with minimum degree at least $D_i$ for parameters\footnote{Notice 
		that the edges $E_t(\bS)$ admit no such type of expander decomposition in our partitioning.} 
		\[
		\eps := \frac{1}{36\log{n}} \quad \text{and} \quad D_i \geq \frac{m_i(\bS)}{36n}.
		\] 
	\end{enumerate}
\end{lemma}
\begin{proof}
	For simplicity of exposition, we drop $(\bS)$  when denoting $E_i(\bS)$'s in the following. 
	We  construct $E_1,\ldots,E_t$ inductively using an auxiliary set of edges $F_0,\ldots,F_t$. Start with $F_0$ being the set of all edges in $graph(\bS)$ and 
	for $i=1$ to $t$ do: 
	\begin{enumerate}[label=$\arabic*).$]
		\item Apply the expander decomposition of~\Cref{prop:expander-decomposition} to $F_{i-1}$ with parameters 
		\[
		\eps = \frac{1}{36\log{n}} \quad \text{and} \quad \dmin = D_i = \frac{\card{F_{i-1}}}{36n},
		\] 
		to get $\eps$-expanders $H^i_1,\ldots,H^i_{k_i}$ each with minimum degree at least $D_i$. 
		\item Let $E_i$ be the union of edges assigned to the expanders in the decomposition of~\Cref{prop:expander-decomposition} in the previous step, and $F_i$ be the leftover edges. 
		Continue to iteration $i+1$. 
	\end{enumerate}
	\noindent
	We argue that $\card{F_i} \leq \card{F_{i-1}}/4$ for all $i \leq t$. For $i > 0$, we have that $\card{F_i}$ is the number of leftover edges of the decomposition
	and thus by~\Cref{prop:expander-decomposition},
	\[
		\card{F_i} \leq 8\eps \cdot \card{F_{i-1}} \cdot \log{n} + n \cdot D_i = \frac{8\card{F_{i-1}}}{36} + \frac{\card{F_{i-1}}}{36} = \frac{\card{F_{i-1}}}{4}.
	\]
	
	Now firstly, $E_i = F_{i-1} \setminus F_{i}$ and so by the above bound, $\card{E_i} \geq 2\card{F_{i}}$. At the same time, $E_{i+1} \subseteq F_{i}$ for and thus $\card{E_{i}} \geq 2 \card{E_{i+1}}$. This proves the first part. 
	
	Secondly, we get property $(ii)$ of the lemma by the choice of $\eps=1/36\log{n}$ in the decomposition and since $D_i = \card{F_{i-1}}/{36n} \geq {\card{E_i}}/{36n}$ as $E_i \subseteq F_{i-1}$. 
\end{proof}

We now have all the tools needed to prove~\Cref{lem:main}. For the rest of the proof, we fix a partitioning $(E_1(\bS),\ldots,E_t(\bS))$ of $graph(\bS)$ using~\Cref{lem:hierarchical} for some $t \geq 1$ such that: 
\begin{align}
	\text{$t$ is the \underline{largest} index where:} \quad m_{t-1}(\bS) \geq n^{1+\delta} \cdot d^{3/2}, \label{eq:choice-t}
\end{align}
where $\delta > 0$ is the absolute constant in~\Cref{lem:main}. 
This means that for every $e \in E_i(\bS)$ for $i < t$, the edge $e$ belongs to some $\eps$-expander $H^i_j$ for $j \in [k_i]$ with min-degree $D_i$ such that, 
\begin{align}
	\eps = \frac{1}{36\log{n}} \qquad \text{and} \qquad D_i \geq \frac{m_i(\bS)}{36n}. \label{eq:can-apply}
\end{align}
This also implies that $D_i \geq (1/36) \cdot d^{3/2} \cdot n^{\delta} \geq (\eps^{-10} \cdot n^{\delta'}) \cdot d$ for some absolute constant $\delta' > 0$ which allows us to apply~\Cref{lem:effective-large} to each expander $H^i_j$ in the proof. 

We now have, 
\begin{align*}
	\text{LHS of~\Cref{lem:main}} &= \Exp_{G,\estar}\Bracket{\min\set{1,\KLD{\bg \cdot \bS \cdot \bB(G)}{\bg \cdot \bS \cdot (\bB(G)+\bB(\estar))}{\bg}}} \\
	&= \sum_{e \in graph(\bS)} \Pr\paren{\estar=e} \cdot  \Exp_{G \mid \estar=e}\Bracket{\min\set{1,\KLD{\bg \cdot \bS \cdot \bB(G)}{\bg \cdot \bS \cdot (\bB(G)+\bB(\estar))}{\bg}}} \tag{whenever $\estar \notin graph(\bS)$, both terms 
	of the KL-divergence will be the same and thus it will be $0$} \\
	&= \sum_{i=1}^{t} \sum_{e \in E_i(\bS)} \Pr\paren{\estar=e} \cdot  \Exp_{G}\Bracket{\min\set{1,\KLD{\bg \cdot \bS \cdot \bB(G)}{\bg \cdot \bS \cdot (\bB(G)+\bB(\estar))}{\bg}}} \tag{by the partitioning of edges of $graph(\bS)$ 
	and since $G \perp \estar$ in $\distspan$} \\
	&=  \frac{1}{{{n}\choose{2}}} \sum_{i=1}^{t} \sum_{e \in E_i(\bS)} \Bracket{\min\set{1,\KLD{\bg \cdot \bS \cdot \bB(G)}{\bg \cdot \bS \cdot (\bB(G)+\bB(e))}{\bg}}}
	\tag{as the marginal distribution of $\estar$ is uniform over ${{V}\choose{2}}$ and we conditioned on $\estar=e$} \\ 
	&\leq \frac{m_t(\bS)}{{{n}\choose{2}}} +  \frac{1}{{{n}\choose{2}}} \cdot \sum_{i=1}^{t-1} \sum_{\substack{e=(u,v)  \in E_i(\bS)}} \Exp_{G}\Bracket{2 \cdot {\eff{G(\bS)+e}{u,v}}} 
	\tag{using the trivial upper bound of $1$ for $E_t(\bS)$ and~\Cref{lem:kl-eff-needed} for $E_1(\bS),\ldots,E_{t-1}(\bS)$} \\
	&= \frac{m_t(\bS)}{{{n}\choose{2}}} +  \frac{2}{{{n}\choose{2}}} \cdot \sum_{i=1}^{t-1}  \sum_{j=1}^{k_i}\sum_{\substack{e=(u,v)  \in H^i_j}} \Exp_{G}\Bracket{{\eff{G(\bS)+e}{u,v}}} 
	\tag{as each $E_i(\bS)$ for $i < t$ is partitioned into expanders $H^i_1,\ldots,H^i_{k_i}$ by~\Cref{lem:hierarchical}} \\
	&=  \frac{m_t(\bS)}{{{n}\choose{2}}} +   \frac{2}{{{n}\choose{2}}} \cdot \sum_{i=1}^{t-1} \sum_{j=1}^{k_i}\sum_{e=(u,v) \in H^i_j} O(\eps^{-4}) \cdot \Paren{\frac{d }{D_i} + \frac{d^3}{D^2_i}} 
	\tag{by~\Cref{lem:effective-large} as each $H^i_j$ is an $\eps$-expander with min-degree $D_i$ (and by~\Cref{eq:can-apply} we can use the lemma)} \\
	&=  \frac{m_t(\bS)}{{{n}\choose{2}}} + \frac{2}{{{n}\choose{2}}} \cdot \sum_{i=1}^{t-1} m_i(\bS) \cdot O(\log^4{n}) \cdot \Paren{\frac{d \cdot n}{m_i(\bS)} + \frac{d^3 \cdot n^2}{m_i(\bS)^2}} \tag{as $\eps = \Theta(1/\log{n})$ and $D_i \geq m_i(\bS)/12n$ 
	by~\Cref{eq:can-apply} and $H^i_1,\ldots,H^i_{k_i}$ have $m_i(\bS)$ edges in total} \\
	%
	%
	&= \frac{m_t(\bS)}{{{n}\choose{2}}} + O(\log^5{n} \cdot \frac{d}{n}) + O(\log^4{n}) \cdot {\frac{d^3}{m_{t-1}(\bS)}} \tag{as $m_i(\bS)$'s decrease (at least) by a geometric series and $t = O(\log{n})$ by~\Cref{lem:hierarchical}} \\
	&\leq \frac{n^{1+\delta} \cdot d^{3/2}}{{{n}\choose{2}}} + O(\log^5{n} \cdot \frac{d}{n}) + O(\log^4{n}) \cdot \frac{d^3}{n^{1+\delta} \cdot d^{3/2}} \tag{by the choice of $t$ in~\Cref{eq:choice-t}} \\
	&= O(\frac{d^{3/2}}{n^{1-\delta}}). 
\end{align*}
This concludes the proof of~\Cref{lem:main}. The next two sections are now dedicated to proving the two main ingredients of this lemma, namely,~\Cref{lem:kl-eff-needed} and~\Cref{lem:effective-large}.

%% file: fig-dist.tex

\begin{tikzpicture}
	
\tikzset{choose/.style={rectangle, draw, rounded corners=5pt, line width=1pt, inner xsep=5pt, inner ysep=2pt]}}
\tikzset{layer/.style={rectangle, rounded corners=4pt, draw, black!10, fill=black!10, opaque=10, line width=0.5pt, inner sep=4pt]}}
\tikzset{vertex/.style={circle, draw, fill=black, line width=1pt, inner sep=1pt]}}

\node[vertex](v11){};
\node[vertex](v12)[above left=10pt and 15pt of v11]{};
\node[vertex](v13)[below left=10pt and 15pt of v11]{};
\begin{scope}[on background layer]
\node[layer, fit=(v11) (v12) (v13)](v1){};
\end{scope}
\node(n1)[below=5pt of v1] {$V_{1}$};

\foreach \i in {2,...,8}
{
	\pgfmathtruncatemacro{\ip}{\i-1};
	\pgfmathtruncatemacro{\xx}{mod(\i,2)};
	\ifthenelse{\xx=0}
	{
	\node[vertex](v\i1)[right=35pt of v\ip1]{};
	\node[vertex](v\i2)[above right=10pt and 15pt of v\i1]{};
	\node[vertex](v\i3)[below right=10pt and 15pt of v\i1]{};
	}
	{
	\node[vertex](v\i1)[right=70pt of v\ip1]{};
	\node[vertex](v\i2)[above left=10pt and 15pt of v\i1]{};
	\node[vertex](v\i3)[below left=10pt and 15pt of v\i1]{};
	}
	\begin{scope}[on background layer]
	\node[layer, fit=(v\i1) (v\i2) (v\i3)](v\i){};
	\end{scope}
	\node(n\i)[below=1pt of v\i] {$V_{\i}$};
}
\foreach \i in {1,...,8}
{
	\foreach \j in {1,2,3}
	{
		\foreach \k in {1,2,3}
		{
			\draw[line width=0.5pt] (v\i\j) to (v\i\k);
		}
	}
}

\foreach \i in {1,...,7}
{
	\pgfmathtruncatemacro{\ip}{\i+1};
	\foreach \j in {1,2,3}
	{
		\foreach \k in {1,2,3}
		{
			\draw[line width=0.5pt] (v\i\j) to (v\ip\k);
		}
	}
}

\node[vertex,red](v22)[above right=10pt and 15pt of v21]{};
\node(n22)[above=5pt of v22]{\textcolor{red}{$\ustar$}};
\node[vertex,red](v73)[below left=10pt and 15pt of v71]{};
\node(n73)[below=5pt of v73]{\textcolor{red}{$\vstar$}};

\draw[line width=1.5pt, dashed, red, bend left=15pt] (v22) to (v73);
\node(n)[above left=5pt and 10pt of v5]{\textcolor{red}{$\estar$}};

\end{tikzpicture}

%% file: kl.tex

\section{KL Divergence Between Sketches of Neighboring Graphs}\label{sec:kl-bounds}

In this section, we prove~\Cref{lem:kl-eff-needed}.
Since the LHS of the inequality is at most one, it suffices to prove the KL-divergence is bounded by the RHS when the effective resistance of $e$ is at most $1/2$.
Thus,~\Cref{lem:kl-eff-needed} is an immediate corollary of the following lemma by setting $G$ to $G(\bS)+e$.
\begin{lemma}\label{lem:kl-eff}
	Let $G$ be a graph and $e=(u,v)$ be an edge in $G$ with effective resistance at most $1/2$. 
	Then we have
	$$
			\KLD{\bg \cdot (\bB(G) - \bB(e))}{\bg \cdot \bB(G)}{\bg} \leq \frac{1}{4} \cdot \eff{G}{u,v},
	$$
	where $\bg$ has independent standard Gaussian coordinates.
\end{lemma}

\bigskip

Now, observe that both $\bg \cdot (\bB(G) - \bB(e))$ and $\bg \cdot \bB(G)$ are high-dimensional Gaussian distributions. 
To prove~\Cref{lem:kl-eff}, we will need the following claim on the KL-divergence between two Gaussian distributions.

\begin{claim}\label{cl:kl-bound}
Let $P$ and $Q$ be $n$-dimensional Gaussian distributions with zero mean and covariance $\bSigma_1\in \R^{n\times n}$ and $\bSigma_2\in \R^{n\times n}$ respectively.
If $\spn(\bSigma_1)=\spn(\bSigma_2)$, which have dimension $k$, then
$$
\KLD{P}{Q}{}=\frac1{2}\left[-\ln \det\left((\bA^+)^{\top}\bSigma_1\bA^+\right)-k+\tr\left((\bA^+)^{\top}\bSigma_1\bA^+\right)\right],
$$
where $\bA^{\top}\bA=\bSigma_2$, and $\bA^+\in\R^{n\times k}$ is the pseudoinverse of $\bA\in\R^{k\times n}$.

\end{claim}
\begin{proof}
	We view both $P$ and $Q$ as random $n$-dimensional column vectors.
	We first apply the linear transformation $(\bA^+)^{\top}$ to both $P$ and $Q$.
	This does not change the KL-divergence between $P$ and $Q$, as it preserves $\frac{\mathbf{d}P}{\mathbf{d}Q}$.
	Note that the covariance matrix $(\bA^+)^{\top}Q$ is
	\[
		\Exp\left[(\bA^+)^{\top}Q\left((\bA^+)^{\top}Q\right)^{\top}\right]=(\bA^+)^{\top}\bSigma_2 \bA^+=(\bA^+)^{\top}\bA^\top \bA \bA^+=\bI_k,
	\]
	where the last equality uses the fact that $\bA$ has linearly independent rows, hence $\bA\bA^{+}=\bI_k$.
	Similarly, the covariance matrix of $(\bA^+)^{\top}P$ is $\bSigma=(\bA^+)^{\top}\bSigma_1 \bA^+$.
	Since $\spn(\bSigma_1)=\spn(\bSigma_2)$, $\bSigma$ is invertible.

	For simplicity of notations, it suffices to consider $\KLD{P}{Q}{}$ for covariance matrices $\bSigma$ and $\bI_k$ respectively.
	By definition, we have
	\begin{align*}
		\KLD{P}{Q}{}&=\E_{x\sim P}\left[\ln \left(\frac{e^{-\frac{1}{2}x^{\top}\bSigma^{-1}x}/\sqrt{(2\pi)^k\det\left(\bSigma\right)}}{e^{-\frac{1}{2}x^{\top}x}/\sqrt{(2\pi)^k}}\right)\right] \\
		&=\E_{x\sim P}\left[-\frac{1}{2}x^{\top}\bSigma^{-1}x-\frac{1}{2}\ln\det\left(\bSigma\right)+\frac{1}{2}\|x\|_2^2\right].
	\end{align*}
	For the first term, since $\bSigma^{-1/2}x$ is a standard Gaussian, $\E_{x\sim P}\left[x^{\top}\bSigma^{-1}x\right]=\E_{x\sim P}\left[\|\bSigma^{-1/2}x\|_2^2\right]=k$.
	For the third term, we have $\E_{x\sim P}\left[\|x\|^2\right]=\sum_i\var\left[x_i\right]=\tr(\bSigma)$.
	We have
	\[
		\KLD{P}{Q}{}=\frac{1}{2}\left(-\ln\det\left(\bSigma\right)-k+\tr(\bSigma)\right).
	\]

	This proves the claim.
\end{proof}

In the later proof, we will need the following claim to bound the logarithm of the determinant term from the previous claim.

\begin{claim}\label{cl:logdet}
For a symmetric matrix $\bA$ with $\left\|\bA\right\|_2\leq 1/2$ one has $\ln \det (\bI+\bA)\geq \tr(\bA)-\tr (\bA^2)$.
\end{claim}
\begin{proof}
	Let $\lambda_1,\ldots,\lambda_n$ be the eigenvalues of $\bA$.
	Then we have
	\begin{align*}
		\ln\det(\bI+\bA)&=\ln\prod_{i=1}^n (1+\lambda_i) 
		=\sum_{i=1}^n \ln(1+\lambda_i) 
		\geq \sum_{i=1}^n (\lambda_i-\lambda_i^2) 
		=\tr(\bA)-\tr(\bA^2),
	\end{align*}
	where the inequality uses the fact that $\ln(1+x)\geq x-x^2$ for all $\left|x\right|\leq 1/2$.
\end{proof}

Now we are ready to prove~\Cref{lem:kl-eff}.

\begin{proof}[Proof of~\Cref{lem:kl-eff}]

We apply~\Cref{cl:kl-bound} with $P=\bg \cdot (\bB(G) - \bB(e))$ and $Q=\mathbf{g}\cdot \bB(G)$, noting that the covariance matrices are 
$\bB(G)^\top\bB(G)-b_e b_e^\top=\bL-b_eb_e^{\top}$ and $\bB(G)^{\top}\bB(G)=\bL$, where $\bL$ is the Laplacian of $G$, $e=\{u, v\}$ and $b_e^{\top}$ is  $e$-th row of $\bB$.
Since $e$ is not a bridge by assumption of the lemma, we have $\spn(\bL-b_eb_e^{\top})=\spn(\bL)$.
We get by~\Cref{cl:kl-bound}, 

\begin{equation}\label{eq:kl-logdet}
\begin{split}
&\quad\,\, \KLD{\bg \cdot (\bB(G) - \bB_{u,v})}{\bg \cdot \bB(G)}{\bg} \\
&=\frac1{2}\left[-\ln \det ((\bA^+)^{\top}(L-b_e b_e^{\top})\bA^+)-k+\tr((\bA^+)^{\top}(L-b_e b_e^{\top})\bA^+)\right]\\
&=\frac1{2}\left[-\ln \det (\bI-w w^{\top})-k+\tr(\bI-w w^{\top}))\right],
\end{split}
\end{equation}
where $k$ is the rank of $\bL$, $\bA$ is a $k\times n$ matrix such that $\bA^{\top}\bA=\bL$, $w=(\bA^+)^{\top}b_e$. Note that 
\[
	\|ww^\top\|_2=\|w\|_2^2=b_e^{\top}\bA^+(\bA^+)^{\top}b_e=b_e^{\top}\bL^+b_e=\eff{G}{u,v}\leq 1/2,
\]
where the last bound is by assumption of the lemma. This means that~\Cref{cl:logdet} applies with $\bA=-ww^{\top}$ and
\begin{equation}\label{eq:logdet}
\ln \det (\bI-ww^{\top})\geq -\tr(ww^{\top})-\tr ((ww^{\top})^2)=-\|w\|_2^2-\|w\|_2^4.
\end{equation}

Substituting into~\eqref{eq:kl-logdet}, we thus get
\begin{equation*}
\begin{split}
\KLD{\bg \cdot (\bB(G) - \bB_{u,v})}{\bg \cdot \bB(G)}{\bg}&=\frac1{2}\left[-\ln \det (\bI-w w^{\top})-k+\tr(\bI-w w^{\top}))\right]\\
&\leq \frac1{2}\left[\|w\|_2^2+\|w\|_2^4 -k+(k-\|w\|_2^2)\right]\\
&=\|w\|_2^4/2\\
&=\eff{G}{u, v}^2/2 \\
&\leq \eff{G}{u, v}/4.
\end{split}
\end{equation*}
as required.
\end{proof}

%% file: eff_resist.tex

\section{Effective Resistance Upper Bound}\label{sec:eff_upper}

In this section, we prove~\Cref{lem:effective-large}.
To this end, let $G$ be a random graph sampled from $\mu(n,d)$ without the edge $(\ustar,\vstar)$ (see~\Cref{fig:dist}).
We will first show that for any expander $H$ with a minimum degree, $G\cap H$ is a \emph{balanced path of expanders} with high probability over the randomness of $G$,\footnote{An edge is in $G\cap H$ if and only if it is in both $G$ and $H$.} then prove that the effective resistance between \emph{every pair} of vertices in such a graph is small.
A balanced path of expanders is defined as follows.

\begin{definition}[Balanced path of expanders]\label{def:path}
We say that a graph $H=(V, E)$ is a balanced length-$d$ path of $\varphi$-expanders if there exists a partition $(V_i)_{i\in [d]}$ of $V$ such that subsets $U_i=V_i\cup V_{i+1}$ for $i\in [d-1]$ defined by
satisfy the following conditions:
\begin{description}
\item[(1)]  for every $i\in [d]$ the graph $H_i=(U_i, E_i)$, $E_i=E\cap (U_i\times U_i)$, induced by $U_i$ is a $\varphi$-expander;
\item[(2)] $\vol(U_i)\leq 3\cdot \vol_{H_j}(U_j)$ for every $i, j\in [d]$;
\item[(3)] for every $i,j\in [d]$ such that $V_j\subset U_i$ (i.e., $j=i$ or $i+1$),
$$
\left|E(V_j,V_j)\right|\geq \frac{1}{8}\vol_{H_i}(U_i).
$$
\end{description}
\end{definition}

Intuitively, a balanced path of expanders consists of a sequence of expanders such that they have roughly the same size, and adjacent expanders have a constant fraction of intersection.
The following lemma states that if $H$ is an expander with a minimum degree, then $H'=G\cap H$ is a balanced length-$d$ path of expanders with high probability.

\begin{lemma}\label{lem:high-prob-path}
	Let $H=(V,E)$ be an $\varepsilon$-expander on $n$ vertices with $m$ edges and minimum degree $\dmin$.
	Let $\delta>0$ be such that $\dmin/d=n^{\delta}$.
	We randomly partition the vertices into $V_1,\ldots,V_d$ such that each vertex is in each $V_i$ with probability $1/d$ independently.
	Then the graph $H'=(V,E')$, where \[
		E'=E\cap \left(\bigcup_{i=1}^d (V_i\times V_i)\cup\bigcup_{i=1}^{d-1}(V_i\times V_{i+1})\right),
	\]
	is a balanced length-$d$ path $\varphi$-expanders with minimum degree $\Omega(\dmin/d)$ and $\Theta(m/d)$ edge with probability $1-n^{-\omega(1)}$, for $\varphi=\Omega(\varepsilon^2\delta^2)$ as long as $n^{\delta}\geq \varepsilon^{-8}\log^{32} n$.
\end{lemma}

The following lemma gives upper bounds the effective resistance of every pair of vertices in a balanced path of expanders.

\begin{lemma}\label{lm:res-diam}
Let $H$ be a balanced path of $\varphi$-expanders $H_1,\ldots, H_d$ as per Definition~\ref{def:path}, and suppose that the minimum degree in $H$ is at least $\dmin$. Then for every pair of distinct vertices $u, v$ in $H$ one has 
$$
\eff{}{u, v}=O\left(\frac1{\varphi^2 \dmin}+\frac{d}{\varphi^2 \vol(U_1)}\right).
$$
\end{lemma}

We will prove the above two lemmas in the following subsections.
Now, we first show that they imply~\Cref{lem:effective-large}.
\begin{proof}[Proof of~\Cref{lem:effective-large}]
	Since $\delta$ is a constant and $\varepsilon\leq 1$, by~\Cref{lem:high-prob-path}, the graph $G\cap H$ is a balanced path of $\varphi$-expanders with probability $1-n^{-\omega(1)}$ for $\varphi=\Omega(\varepsilon^2)$.
	Since $H$ is a nonempty graph with minimum degree $D$, it must have at least $\Omega(D^2)$ edges.
	Thus,~\Cref{lem:high-prob-path} implies that $G\cap H$ has $\Omega(D^2/d)$ edges with minimum degree $\Omega(D/d)$.
	Combining the bounds with~\Cref{lm:res-diam}, the effective resistance between $u$ and $v$ (for $(u,v)=e$) in $G\cap H$ is
	\[
		\eff{G\cap H}{u, v}=O\left(\varepsilon^{-4}\cdot \left(\frac{d}{D}+\frac{d^3}{D^2}\right)\right),
	\]
	with probability $1-n^{-\omega(1)}$.

	Finally, observe that $G\cap H$ is a subgraph of $G(\bS)+e$, we have $\eff{G(\bS)+e}{u,v}\leq \eff{G\cap H}{u, v}$, and that edge $e$ is in $G(\bS)+e$, we have $\eff{G(\bS)+e}{u,v}\leq 1$.
	It follows that
	\[
		\E_{G}\left[\eff{G(\bS)+e}{u,v}\right]\leq O\left(\varepsilon^{-4}\cdot \left(\frac{d}{D}+\frac{d^3}{D^2}\right)\right).
	\]
	This proves the lemma.
\end{proof}

\subsection{$G\cap H$ is a balanced path of expanders}\label{sec:high-prob-path}
We first prove~\Cref{lem:high-prob-path}.
The proof consists of two parts: every $U_i$ is a $\varphi$-expander (condition {\bf(1)}); the volumes of sets are concentrated (condition {\bf(2), (3)}, minimum degree, the number of edges in $G(\bS)$).
We state the two part in the following two lemmas respectively.
\begin{lemma}\label{lem:sample_expander}
	Let $H$ be an $\varepsilon$-expander on $n$ vertices with minimum degree $\dmin$.
	Fix $p\in(0,1)$, and $\delta>0$ such that $\dmin\cdot p\geq n^{\delta}$.
	Let $H_{\smp}$ be resulting graph after sampling each vertex in $H$ independently with probability $p$, then $H_{\smp}$ is an $\Omega(\varepsilon^2\delta^2)$-expander with probability at least $1-n^{-\omega(1)}$, as long as $n^{\delta}\geq \varepsilon^{-8}\cdot \log^{32} n$.
\end{lemma}
The proof of the lemma is deferred to~\Cref{sec:expander_sample}.
\begin{lemma}\label{lem:sample_volume}
	Let $H'$ be the random graph defined as in~\Cref{lem:high-prob-path}.
	The with probability $1-n^{-\omega(1)}$, we have for every $\left|i-j\right|\leq 1$, 
	\[
		\frac{49}{64}\cdot \frac{2m}{d^2}\leq \left|\left\{u\in V_i,v\in V_j: (u,v)\in E\right\}\right|\leq \frac{81}{64}\cdot \frac{2m}{d^2},
	\]
	and the minimum degree of $H'$ is at least $\frac{7}{8}\cdot \frac{\dmin}{d}$.
\end{lemma}

\begin{proof}
	Consider a part $V_i$.
	Since each vertex is in $V_i$ with probability $1/d$ independently, the degree (in $H$) of all vertices that belong to $V_i$ is $2m/d$.
	Moreover, by Bernstein's inequality, we have 
	\begin{align*}
		\Pr\left[\sum_{u\in V_i} d_u\geq \frac{9}{8}\cdot \frac{2m}{d}\right]&\leq \exp\left(-\Omega\left(\frac{(m/d)^2}{\sum_{u} d_u^2/d+n\cdot (m/d)}\right)\right) \\
		&\leq \exp\left(-\Omega\left(\frac{(m/d)^2}{\dmin\cdot \sum_{u} d_u/d+n\cdot (m/d)}\right)\right) \\
		&\leq \exp\left(-\Omega\left(m/nd\right)\right) \\
		&\leq \exp\left(-\Omega\left(\dmin/d\right)\right) \\
		&\leq n^{-\omega(1)},
	\end{align*}
	where $d_u$ is the degree of vertex $u$ in $H$.
	Similarly, we have
	\[
		\Pr\left[\sum_{u\in V_i} d_u\leq \frac{7}{8}\cdot \frac{2m}{d}\right]\leq n^{-\omega(1)}.
	\]

	Next, fix a vertex $u$, and consider the number of its neighbors in $H$ that belong to part $V_j$.
	The expected number of such neighbors is $d_u/d$.
	Again by Bernstein's inequality, we have
	\begin{align*}
		\Pr\left[\left|E_H(u, V_j)\right|\geq \frac{9}{8}\cdot \frac{d_u}{d}\right]&\leq \exp\left(-\Omega\left(\frac{(d_u/d)^2}{\sum_{v:(u,v)\in E}1/d+d_u/d}\right)\right) \\
		&=\exp\left(-\Omega(d_u/d)\right) \\
		&\leq \exp\left(-\Omega\left(\dmin/d\right)\right) \\
		&=n^{-\omega(1)},
	\end{align*}
	and
	\[
		\Pr\left[\left|E_H(u, V_j)\right|\leq \frac{7}{8}\cdot \frac{d_u}{d}\right]\leq n^{-\omega(1)}.
	\]
	Note that the above bound on the probability still holds even if we condition on $u\in V_j$.
	Hence, the degree of $u$ in $H'$ is at least $\frac{7}{8}\cdot \frac{\dmin}{d}$, proving the minimum degree bound.

	Next, by union bound, we have $\sum_{u\in V_i} d_u\in \left[\frac{7}{8}\cdot \frac{2m}{d},\frac{9}{8}\cdot \frac{2m}{d}\right]$ and $\left|E_H(u, V_j)\right|\in \left[\frac{7}{8}\cdot \frac{d_u}{d},\frac{9}{8}\cdot \frac{d_u}{d}\right]$ for all $u,i,j$ with probability $1-n^{-\omega(1)}$.
	When it happens, for any $\left|i-j\right|\leq 1$, we have
	\[
		\frac{49}{64}\cdot \frac{2m}{d^2}\leq \left|\left\{u\in V_i,v\in V_j: (u,v)\in E\right\}\right|\leq \frac{81}{64}\cdot \frac{2m}{d^2}.
	\]
	This proves the lemma.
\end{proof}

\begin{proof}[Proof of~\Cref{lem:high-prob-path}]
	For condition {\bf(1)}, since each $U_i$ (marginally) is formed by including each vertex of $H$ with probability $2/d$ (or $1/d$ if $i=1$), by~\Cref{lem:sample_expander}, $U_i$ is an $\varphi$-expander with probability $1-n^{-\omega(1)}$.

	The minimum degree bound follows from~\Cref{lem:sample_volume}.
	Moreover, it implies that with probability $1-n^{-\omega(1)}$, $\vol_{H'}(V_i)\in \left[3\cdot \frac{49}{64}\cdot \frac{2m}{d^2}, 3\cdot \frac{81}{64}\cdot \frac{2m}{d^2}\right]$ for $i\neq 1,d$, and $\vol_{H'}(V_i)\in \left[2\cdot \frac{49}{64}\cdot \frac{2m}{d^2}, 2\cdot \frac{81}{64}\cdot \frac{2m}{d^2}\right]$ for $i=1$ or $d$.
	Thus, condition {\bf(2)} and {\bf(3)} follow.
	Finally, the total number of edges in $H'$ is $\frac{1}{2}\sum_{i=1}^d\vol_{H'}(V_i)=\Theta(m/d)$.
	This proves the lemma.
\end{proof}

\subsection{Effective resistance upper bound for balanced path of expanders}\label{sec:res-diam}


To prove the effective resistance upper bound, we first prove the following claim on the expansion of cuts in $H$.
\begin{claim}\label{cl:expansion-lb}
For every $S\subseteq V$ such that $\vol(S)\leq \frac1{2}\vol(V)$ one has 
$$
|E(S, V\setminus S)|\geq \Omega(\varphi)\cdot \min\{\vol(S), \vol(U_1)\}
$$
\end{claim}
\begin{proof}
We write $\vol_{H_i}(S)$ to denote the sum of degrees of vertices in $S$ in $H_i$, and $\vol(S)$ to denote the sum of degrees of vertices in $S$ in $H$.

First suppose that $\vol_{H_i}(S\cap U_i)\leq (7/8)\vol_{H_i}(U_i)$ for all $i\in [d]$. Since every $H_i$ is a $\varphi$-expander by assumption, and every edge of $H$ belongs to at most two of the $H_i$'s (by property {\bf (1)} in Definition~\ref{def:path}), we have
\begin{equation*}
\begin{split}
|E(S, V\setminus S)|&\geq \frac1{2}\sum_{i\in [d]} |E_i(S\cap U_i, U_i\setminus S)|\\
&\geq (\varphi/2)\cdot\sum_{i\in [d]} \min\{\vol_{H_i}(U_i\cap S), \vol_{H_i}(U_i\setminus S)\}\\
&\geq \Omega(\varphi/2)\cdot\sum_{i\in [d]} \vol_{H_i}(U_i\cap S)\\
&\geq\Omega(\varphi)\cdot \vol(S).
\end{split}
\end{equation*}

Now suppose that $\vol_{H_i}(S\cap U_i)> (7/8) \vol_{H_i}(U_i)$ for some $i\in [d]$.
Then there must exist some $i,j\in [d]$ such that $\left|i-j\right|=1$, $\vol_{H_i}(S\cap U_i)>  (7/8) \vol_{H_i}(U_i)$, and $\vol_{H_j}(S\cap U_j)\leq  (7/8) \vol_{H_j}(U_j)$.
This is because otherwise, we would have $\vol_{H_i}(S\cap U_i)> (7/8) \vol_{H_i}(U_i)$ for all $i\in [d]$, and
$$
\vol(V\setminus S)\leq \sum_{j\in [d]} \vol_{H_j}(U_j\setminus S)\leq (1/8) \sum_{j\in [d]} \vol_{H_j}(U_j)\leq (1/4)\vol(V),
$$
a contradiction with the assumption that $\vol(S)\leq (1/2)\vol(V)$.

Without loss of generality assume that $j=i+1$.
Then we have 
\begin{align*}
	\vol_{H_j}(S\cap U_j)&\geq 2\left|E(S\cap V_j,S\cap V_j)\right| \\
	&=\vol_{H_i}(S\cap U_i)-2\left|E(S\cap V_i,S\cap V_i)\right|-2\left|E(S\cap V_i,S\cap V_j)\right| \\
	&\geq \frac{7}{8}\vol_{H_i}(U_i)-2\left|E(V_i,V_i)\right|-2\left|E(V_i,V_j)\right| \\
	&=\frac{7}{8}\vol_{H_i}(U_i)-\left(\vol_{H_i}(U_i)-2\left|E(V_j,V_j)\right|\right) \\
	&=2\left|E(V_j,V_j)\right|-\frac{1}{8}\vol_{H_i}(U_i) \\
	&\geq \frac{1}{8}\vol_{H_i}(U_i),
\end{align*}
where the last inequality uses property $\bf(3)$ of balanced path of expanders. 

Then by property $\bf(2)$, we have $\frac{1}{24}\vol_{H_j}(U_j)\leq \vol_{H_j}(S\cap U_j)\leq \frac{7}{8}\vol_{H_j}(U_j)$, and by the expansion property in $H_j$, we have
\begin{align*}
&\quad\,\,|E(S, V\setminus S)| \\
&\geq |E_j(S\cap U_j,  U_j\setminus S)|\\
&\geq \varphi\cdot \min\{\vol_{H_j}(S), \vol_{H_j}(U_j\setminus S)\} \\
&\geq \varphi\cdot (1/24)\vol_{H_j}(U_j) \\
&\geq \Omega(\varphi) \vol(U_1).
\end{align*}
This proves the lemma.
\end{proof}

We are ready to prove~\Cref{lm:res-diam}.

\begin{proof}[Proof of~\Cref{lm:res-diam}]
Fix a pair of  distinct vertices $u, v\in V$. Let $f\in \R^E$ be the unit electrical flow from  $u$ to $v$, and let $\phi\in \R^V$ be the corresponding vector of potentials.
Recall that the flow on an edge $(a, b)$ satisfies $f_{ab}=\phi_a-\phi_b$, and that the effective resistance between $u$ and $v$ satisfies $\eff{}{u, v}=\phi_u-\phi_v$.

We define a sequence of thresholds $\theta_j, j\geq 0$ and define
$$
L_j=\{w\in V: \phi_w\geq \theta_j\}.
$$

Let $\theta_0=\phi_u$, so that $L_0$ contains $u$. Fix $j$ such that $v\notin L_j$. Note that since $u\in L_j, v\in V\setminus L_j$, we have that the total electrical flow across the cut $(L_j, V\setminus L_j)$ is one. Since for every edge $(a, b)$ the flow on $(a, b)$ is $\phi_a-\phi_b$, we get
\begin{equation}
\sum_{(a, b)\in E: a\in L_j, b\in V\setminus L_j} (\phi_a-\phi_b)=1.
\end{equation}
By the definition of $L_j$, we must have $\phi_a>\phi_b$ for all $(a, b)\in E$ and $a\in L_j,b\in V\setminus L_j$.
Thus, at most $\frac1{2}|E(L_j, V\setminus L_j)|$ such edges $(a, b)$ can have $\phi_a-\phi_b\geq \frac{2}{|E(L_j, V\setminus L_j)|}$.

This allows us to define 
$$
\theta_{j+1}=\theta_j-\frac{2}{|E(L_j, V\setminus L_j)|}.
$$
Then for all edges $(a, b)\in E(L_j, V\setminus L_j)$ such that $\phi_a-\phi_b<\frac{2}{|E(L_j, V\setminus L_j)|}$, $b$ must be in $L_{j+1}$.
Therefore,
$$
\vol(L_{j+1})\geq \vol(L_{j})+\frac1{2}|E(L_j, V\setminus L_j)|.
$$

We now lower bound $|E(L_j, V\setminus L_j)|$.  We have by Claim~\ref{cl:expansion-lb} that 
$$
|E(L_j, V\setminus L_j)|\geq \Omega(\varphi)\cdot \min\{\vol(L_j), \vol(U_1)\}, 
$$
and therefore
$$
\vol(L_{j+1})\geq \vol(L_{j})+\Omega(\varphi)\cdot \min\{\vol(L_{j}), \vol(U_1)\}.
$$

Let $j^*$ be the smallest such that $\vol(L_{j^*})>\vol(U_1)$. For all $j\leq j^*$ we have
\begin{equation*}
\begin{split}
\vol(L_j)&\geq \vol(L_{j-1})+\Omega(\varphi)\cdot \vol(L_{j-1})\\
&=(1+\Omega(\varphi)) \vol(L_{j-1})\\
&\geq (1+\Omega(\varphi))^j \vol(L_0)\\
&\geq (1+\Omega(\varphi))^j \dmin
\end{split}
\end{equation*}
and 
\begin{equation*}
\begin{split}
|E(L_{j}, V\setminus L_{j})|&\geq \Omega(\varphi) \vol(L_{j})\\
&\geq \Omega(\varphi) (1+\Omega(\varphi))^j \dmin.
\end{split}
\end{equation*}
This means in particular that 
\begin{equation*}
\begin{split}
\theta_{j^*}-\theta_0 &\leq \sum_{j<j^*} \frac{2}{\Omega(\varphi) \vol(L_{j})}\\
&\leq \sum_{j\geq 0} \frac{O(1)}{\varphi \dmin} (1+\Omega(\varphi))^{-j} \\
&=\frac{O(1)}{\varphi^2 \dmin}.
\end{split}
\end{equation*}

Let $J$ be the smallest such that $\vol(L_j)> \frac{1}{2}\vol(V)$.
We now consider $j^*<j\leq J$. Here we have 
$$
|E(L_{j-1}, V\setminus L_{j-1})|\geq \Omega(\varphi) \vol(U_1).
$$
Since 
$$
\vol(L_{j})\geq \vol(L_{j-1})+\frac1{2}|E(L_{j-1}, V\setminus L_{j-1})|\geq  \vol(L_{j})+\Omega(\varphi) \cdot \vol(U_1),
$$
we must have $J=O(\vol(V)/(\varphi\cdot \vol(U_1)))=O(d/\varphi)$ (recall that $d$ is the number of layers).  Recalling that 
$$
\theta_{j}=\theta_{j-1}+\frac{2}{|E(L_{j-1}, V\setminus L_{j-1})|},
$$
we now get that 
$$
\theta_J-\theta_{j^*}\leq O(d/\varphi)\cdot O\left(\frac1{\varphi\cdot \vol(U_1)}\right)=O\left(\frac{d}{\varphi^2 \vol(U_1)}\right). 
$$

Putting the above bounds together, we have
\[
	\theta_J-\theta_0\leq O\left(\frac1{\varphi^2 \dmin}+\frac{d}{\varphi^2 \vol(U_1)}\right).
\]
That is, by the definition of $J$, the volume of all vertices $w$ with $$\phi_u-\phi_w\leq c\cdot \left(\frac1{\varphi^2 \dmin}+\frac{d}{\varphi^2 \vol(U_1)}\right)$$ (for a sufficiently large constant $c$) is more than $\frac{1}{2}\vol(V)$.

By a symmetric argument from vertex $v$, we get that the volume of all vertices $w$ with $$\phi_w-\phi_v\leq c\cdot \left(\frac1{\varphi^2 \dmin}+\frac{d}{\varphi^2 \vol(U_1)}\right)$$ (for a sufficiently large constant $c$) is also more than $\frac{1}{2}\vol(V)$.
In particular, there exists one $w$ such that both $\phi_u-\phi_w$ and $\phi_w-\phi_v$ are at most $c\cdot \left(\frac1{\varphi^2 \dmin}+\frac{d}{\varphi^2 \vol(U_1)}\right)$, implying that
\[
	\phi_u-\phi_v\leq O\left(\frac1{\varphi^2 \dmin}+\frac{d}{\varphi^2 \vol(U_1)}\right).
\]

Hence,
$$
\eff{}{u, v}=O\left(\frac1{\varphi^2 \dmin}+\frac{d}{\varphi^2 \vol(U_1)}\right), 
$$
as required.
\end{proof}

%% file: expander_sample.tex





\section{Vertex-Sample Expanders}\label{sec:expander_sample}

	In this section, we prove~\Cref{lem:sample_expander}.
	To prove that the subsampled graph $H_{\smp}$ is also an expander, instead of working with conductance, we will apply Cheeger's inequality (\Cref{lem:conduct_spgap}), and work with the spectral gap.
	
	More specifically, for expander $H$, we first focus on the matrix $\bM=\bI-\frac{1}{2}\widetilde{\mathbf{L}}$.
	From the spectral gap of $H$ and the properties of the normalized Laplacian matrix, we know that the largest eigenvalue of $\bM$ is $1$, and all other eigenvalues are nonnegative and bounded away from $1$.
	To prove that $H_{\smp}$ is also an expander, we will study the matrix $\bM_{\smp}$ defined similarly.
	The idea is to prove that for some large integer $t$, $\tr(\bM_{\smp}^t)$ is close to $1$ with high probability.
	This gives an estimate on the sum of the $t$-th power of eigenvalues of $\bM_{\smp}$.
	In particular, since its largest eigenvalue is still $1$, the $t$-th power of the second largest eigenvalue must be very small, implying a non-trival spectral gap.

	The key lemma in this argument is~\Cref{lem_moment_trace}, which proves an upper bound on the $q$-th moment of $\tr(\oM_{\smp}^t)$ ($\oM_{\smp}$ is an approximation of $\bM_{\smp}$ that is easier to work with).
	By applying the standard argument for obtaining concentration from moment bounds, we prove that $\tr(\oM_{\smp}^t)$ (and hence, $\tr(\bM_{\smp}^t)$) is small with very high probability.

	\begin{proof}[Proof of~\Cref{lem:sample_expander}]
		Since $H$ is an $\varepsilon$-expander, by~\Cref{lem:conduct_spgap}, the spectral gap of $H$ is at least $\varepsilon^2/2$.
		That is, let $\bA$ be the adjacency matrix of $H$ and $\bD$ be its degree matrix, the top eigenvalue of the matrix
		\[
			\bM=\bI-\frac{1}{2}\cdot \widetilde{\mathbf{L}}=\frac{1}{2}\left(\bI+\bD^{-1/2}\bA\bD^{-1/2}\right)
		\]
		is equal to $1$, and all other eigenvalues are between $0$ and $1-\varepsilon^2/4$.

		Let $S$ be the set of vertices that are present in $H_{\smp}$, and let $\bA_{\smp}$ be the $n\times n$ matrix obtained by zeroing all columns and rows of $\bA$ corresponding to $\overline{S}$.
		Equivalently, $\bA_{\smp}$ is the adjacency matrix of $H_{\smp}$ embedded in an $n\times n$ all-zero matrix.
		Let $\bD_{\smp}$ be its degree matrix.
		It suffices to upper bound the second largest eigenvalue of
		\[
			\bM_{\smp}=\frac{1}{2}\left(\bI_S+\bD_{\smp}^{-1/2}\bA_{\smp}\bD_{\smp}^{-1/2}\right),
		\]
		where $\bI_S$ is the $n\times n$ matrix with ones only in the diagonal entries corresponding to $S$ (note that the second largest eigenvalue is the same as the smaller matrix with only rows and columns corresponding to $S$).
		To this end, we will first focus on the matrix
		\[
			\oM_{\smp}=\frac{1}{2}\left(\bI_S+p^{-1}\cdot \bD^{-1/2}\bA_{\smp}\bD^{-1/2}\right),
		\]
		where we replaced the degree matrix $\bD_{\smp}$ by the expected degree $p\cdot \bD$, then show that it is close to $\bM_{\smp}$ with high probability.

		\bigskip

		To bound the second largest eigenvalue of $\oM_{\smp}$, we analyze the trace of $\oM_{\smp}^t$ for some integer $t$, and apply the fact that the trace is equal to the sum of eigenvalues.
		We have the following moment bound for the trace.


		\begin{lemma}\label{lem_moment_trace}
			For any $t\geq \frac{256\ln n}{\varepsilon^2\delta}$, $t^4\leq p\cdot \dmin$, $q\geq 1$ and $qt^2\leq p\cdot \dmin$, we have $$\E\left[{\tr(\oM_{\smp}^t)}^q\right]\leq \E\left[{\tr(\oM_{\smp}^t)}^{q-1}\right]\cdot \left(\tr(\bM^t)+\frac{(q+4)t^3}{p\cdot \dmin}\right).$$
		\end{lemma}
		


		Recall that in Lemma~\ref{lem_moment_trace} above $\delta>0$ is a parameter such that $\dmin\cdot p\geq n^\delta$. We defer the proof of Lemma~\ref{lem_moment_trace} to the end of this section.
		Now fix $t$ to be any integer between $\frac{256\ln n}{\varepsilon^2\delta}$ and $(p\cdot\dmin)^{1/4}$.
		Since $\tr(\bM^t)\leq 1+n\cdot (1-\varepsilon^2/2)^t\leq 1+n^{-32}$, we have
		\[
			\E\left[{\tr(\oM_{\smp}^t)}^q\right]\leq \left(1+n^{-32}+(q+4)\cdot n^{-\delta/4}\right)^{q}\leq e^{2q^2\cdot n^{-\delta/4}},
		\]
		for any $q\geq 5$.
		By setting $q=n^{\delta/8}\cdot \log n$, and by Markov's inequality, we have
		\[
			\Pr\left[\tr(\oM_{\smp}^t)^q>e^{4q^2\cdot n^{-\delta/4}}\right]<n^{-\omega(1)}.
		\]

		It implies that
		\[
			\Pr\left[\tr(\oM_{\smp}^t)\geq 1+8n^{-\delta/8}\cdot\log n\right]\leq n^{-\omega(1)}.
		\]

		Since $H$ has minimum degree $\dmin$, by Chernoff bound, all vertices $v$ in $H_{\smp}$ have degree at least
		\[
			d_{\smp,v}\geq (1-(p\cdot\dmin)^{-1/2}\log n)p\cdot d_v
		\]
		with probability at least $1-n^{-\omega(1)}$.
		This implies that for all edges $(u, v)$ such that $u,v\in S$,
		\begin{align*}
			(\bM_{\smp})_{u, v}&=\frac{1}{2}\cdot \frac{1}{d_{\smp,u}^{1/2}\cdot d_{\smp,v}^{1/2}} \\
			&\leq \frac{1}{2}\cdot \frac{1}{(1-(p\cdot \dmin)^{-1/2}\log n)p\cdot d_u^{1/2}\cdot d_v^{1/2}} \\
			&\leq (1+2(p\cdot\dmin)^{-1/2}\log n)(\oM_{\smp})_{u, v}.
		\end{align*}

		Hence,
		\begin{align*}
			\tr(\bM_{\smp}^t)&\leq (1+2(p\cdot\dmin)^{-1/2}\log n)^t\cdot \tr(\oM_{\smp}^t)\\
			&\leq (1+4t(p\cdot\dmin)^{-1/2}\log n)\cdot \tr(\oM_{\smp}^t),
		\end{align*}
		with probability at least $1-n^{-\omega(1)}$.


		Therefore, by union bound, $\tr(\bM_{\smp}^t)\leq 1+n^{-\delta/16}$ with probability at least $1-n^{-\omega(1)}$.
		Since the largest eigenvalue of $\bM_{\smp}$ is equal to $1$, it implies that the second largest eigenvalue of $\bM_{\smp}^t$ is at most $n^{-\delta/16}$ with probability at least $1-n^{-\omega(1)}$.
		In this case, the second largest eigenvalue of $\bM_{\smp}$ is at most
		\[
			\left(n^{-\delta/16}\right) ^{1/t}\leq 1-\Omega(\varepsilon^2\delta^2).
		\]
		This implies that the spectral gap of $H_{\smp}$ is at least $\Omega(\varepsilon^2\delta^2)$.
		Another application of~\Cref{lem:conduct_spgap} proves the lemma.
	\end{proof}

	To prove Lemma~\ref{lem_moment_trace}, we will use the following fact about expanders.

	\begin{lemma}\label{lem_power_u_v}
		For any $u,v\in[n]$, we have
		\[
			\left(\bM^k\right)_{u,v}\leq \sqrt{d_ud_v}\cdot \left(\frac{1}{D}+\left(\frac{1}{\dmin}-\frac{1}{D}\right)\cdot e^{-\varepsilon^2 k/4}\right),
		\]
		where $D=\sum_{u} d_u$.
	\end{lemma}

	\begin{proof}
		$\bM$ has top eigenvalue $1$ with eigenvector $\bD^{1/2}\bOne$, and all its other eigenvalues are in $[0, 1-\varepsilon^2/4]$.
		For any indicator vector $\be_v$ for $v\in [n]$ and $k\geq 0$, by the fact that $\left<\bD^{1/2}\bOne, \be_v-\frac{d_v^{1/2}}{D}\cdot \bD^{1/2}\bOne\right>=0$, we have
		\begin{align*}
			\left\|\bM^k \left(\be_v-\frac{d_v^{1/2}}{D}\cdot \bD^{1/2}\bOne\right)\right\|_2&\leq \left\|\be_v-\frac{d_v^{1/2}}{D}\cdot \bD^{1/2}\bOne\right\|_2\cdot (1-\varepsilon^2/4)^k \\
			&=\sqrt{1-\frac{d_v}{D}}\cdot (1-\varepsilon^2/4)^k.
		\end{align*}
		Therefore, for any $u,v\in[n]$, we have
		\begin{align*}
			\left(\bM^k\right)_{u,v}&=\be_u^{\top} \bM^k \be_v \\
			&=\frac{\sqrt{d_ud_v}}{D}+\left(\be_u-\frac{d_u^{1/2}}{D}\cdot \bD^{1/2}\bOne\right)^{\top}\bM^k\left(\be_v-\frac{d_v^{1/2}}{D}\cdot \bD^{1/2}\bOne\right) \\
			&\leq \frac{\sqrt{d_ud_v}}{D}+\left\|\be_u-\frac{d_u^{1/2}}{D}\cdot \bD^{1/2}\bOne\right\|_2\cdot \left\|\bM^k\left(\be_v-\frac{d_v^{1/2}}{D}\cdot \bD^{1/2}\bOne\right)\right\|_2 \\
			&\leq \frac{\sqrt{d_ud_v}}{D}+\sqrt{1-\frac{d_u}{D}}\cdot\sqrt{1-\frac{d_v}{D}}\cdot (1-\varepsilon^2/4)^k \\
			&\leq \sqrt{d_ud_v}\cdot \left(\frac{1}{D}+\left(\frac{1}{\dmin}-\frac{1}{D}\right)\cdot e^{-\varepsilon^2 k/4}\right).
		\end{align*}
	\end{proof}

	Now, we are ready to prove Lemma~\ref{lem_moment_trace}.

	\begin{proof}[Proof of Lemma~\ref{lem_moment_trace}]
		Observe that
		\begin{align*}
			\E\left[{\tr(\oM_{\smp}^t)}^q\right]=\sum_{\substack{u_{1,1},\ldots,u_{1,t},\\\ldots\\u_{q,1},\ldots,u_{q,t}\in[n]}}\E\left[\prod_{i=1}^q\prod_{j=1}^t\left(\oM_{\smp}\right)_{u_{i,j},u_{i,j+1}}\right],
		\end{align*}
		where $u_{i,t+1}=u_{i,1}$ for every $i=1,\ldots, q$.
		For simplicity of notation, let $\bU=(U_1,\ldots,U_q)$ denote $(u_{1,1},\ldots,u_{q,t})$, where $U_i=(u_{i,1},\ldots,u_{i,t})$.
		We view each $U_i$ as a cycle, and hence, $u_{i,t+j}=u_{i,j}$ for all $j$.

		If both $u_{i,j}$ and $u_{i,j+1}$ are sampled, we have
		\begin{itemize}
			\item if $u_{i,j}\neq u_{i,j+1}$ (a non-self-loop), $\left(\oM_{\smp}\right)_{u_{i,j},u_{i,j+1}}=p^{-1}\cdot \bM_{u_{i,j},u_{i,j+1}}$;
			\item if $u_{i,j}=u_{i,j+1}$ (a self-loop), $\left(\oM_{\smp}\right)_{u_{i,j},u_{i,j+1}}=\bM_{u_{i,j},u_{i,j+1}}$.
		\end{itemize}
		If either of them is not sampled, then $\left(\oM_{\smp}\right)_{u_{i,j},u_{i,j+1}}=0$.
		Thus, we have
		\begin{align*}
			&\, \sum_{U_1,\ldots,U_{q-1},U_q}\E\left[\prod_{i=1}^q\prod_{j=1}^t\left(\oM_{\smp}\right)_{u_{i,j},u_{i,j+1}}\right] \\
			&=\sum_{U_1,\ldots,U_{q-1}}\E\left[\prod_{i=1}^{q-1}\prod_{j=1}^t\left(\oM_{\smp}\right)_{u_{i,j},u_{i,j+1}}\cdot \sum_{u_{q,1},\ldots,u_{q,t}}\prod_{j=1}^t\left(\oM_{\smp}\right)_{u_{q,j},u_{q,j+1}}\right] \\
			&=\sum_{U_1,\ldots,U_{q-1}}\E\left[\prod_{i=1}^{q-1}\prod_{j=1}^t\left(\oM_{\smp}\right)_{u_{i,j},u_{i,j+1}}\right]\cdot  \sum_{u_{q,1},\ldots,u_{q,t}}\frac{p^{\left|U_q\setminus \left(U_1\cup\cdots\cup U_{q-1}\right)\right|}}{p^{\left|\left\{j: u_{q,j}\neq u_{q,j+1}\right\}\right|}}\cdot \prod_{j=1}^t\bM_{u_{q,j},u_{q,j+1}}.
		\end{align*}

		Since $\oM_{\smp}$ has nonnegative entries, it suffices to prove that
		\begin{equation}\label{eqn_one_cycle}
			\sum_{u_{q,1},\ldots,u_{q,t}}\frac{p^{\left|U_q\setminus \left(U_1\cup\cdots\cup U_{q-1}\right)\right|}}{p^{\left|\left\{j: u_{q,j}\neq u_{q,j+1}\right\}\right|}}\cdot \prod_{j=1}^t\bM_{u_{q,j},u_{q,j+1}}\leq \tr(\bM^t)+\frac{(q+4)t^3}{p\cdot \dmin},
		\end{equation}
		for any $U_1,\ldots,U_{q-1}$.

		First note that if $U_q$, excluding all self-loops, is not a cycle in $H$, then the corresponding product is zero.
		In the following, we will only consider cycles $U_q$.
		To bound upper bound the LHS of~\eqref{eqn_one_cycle}, we divide the set of all $U_q$ into three cases, and take the sums correspondingly.

		\paragraph{Case 1.} We first consider all $U_q$ such that
		\begin{itemize}
			\item $U_q\cap (U_1\cup\cdots\cup U_{q-1})=\emptyset$;
			\item all edges $(u_{q,j},u_{q,j+1})$ with $u_{q,j}\neq u_{q,j+1}$ have a different $u_{q,j+1}$. 
		\end{itemize}
		That is, $U_q$, removing all self-loops, is a simple cycle disjoint from $U_1,\ldots,U_{q-1}$.

		We observe that for such $U_q$,
		\[
			\left|\left\{j:u_{q,j}\neq u_{q,j+1}\right\}\right|\leq \left|U_q\right|=\left|U_q\setminus (U_1\cup\cdots\cup U_{q-1})\right|.
		\]
		Thus, if we take the sum only over all such $U_q$, the LHS of~\eqref{eqn_one_cycle} is at most
		\[
			\sum_{u_{q,1},\ldots,u_{q,t}}\prod_{j=1}^t \bM_{u_{q,j},u_{q,j+1}}=\tr(\bM^t).
		\]
		\paragraph{Case 2.} Next, we consider all $U_q$ such that
		\begin{itemize}
			\item $U_q\cap(U_1\cup\cdots\cup U_{q-1})\neq\emptyset$.
		\end{itemize}
		Consider any such $U_q$, let $j^*=j^*(U_q)$ be the smallest index in $[t]$ such that $u_{q,j^*-1}\neq u_{q,j^*}$ and $u_{q,j^*}\in U_1\cup\cdots\cup U_{q-1}$.
		For $j\in[t]$, 
		\begin{itemize}
			\item if $u_{q,j-1}\neq u_{q,j}$ and $u_{q,j}\notin U_1\cup\cdots\cup U_{q-1}$, let $\fir(j)$ be first index, in the list $(j^*,j^*+1,\ldots,t,1,\ldots,j^*-1)$, such that $u_{q,j}=u_{q,\fir(j)}$ (i.e., the first after $j^*$ in the cyclic order);
			\item if $u_{q,j-1}\neq u_{q,j}$ and $u_{q,j}\in U_1\cup\cdots\cup U_{q-1}$, let $\fir(j)={\uparrow}$ (indicating it is in the previous cycles);
			\item if $u_{q,j-1}=u_{q,j}$, let $\fir(j)=\bot$.
		\end{itemize}
		
		To bound the sum over all such $U_q$, we will group the terms according to $\fir(\cdot)$ and $j^*$.
		We first take the sum over all $U_q$ with the same $\fir(\cdot)$ and $j^*$, then sum over all possible $\fir(\cdot)$ and $j^*$.
		Now fix $\fir(\cdot)$ and $j^*$, we use the following bounds for each factor $\bM_{u_{q,j},u_{q,j+1}}$.
		\begin{itemize}
			\item if $\fir(j+1)=\bot$, $\bM_{u_{q,j},u_{q,j+1}}=\frac{1}{2}=\frac{1}{2}\cdot \frac{d_{u_{q,j}}^{1/2}}{d_{u_{q,j+1}}^{1/2}}$;
			\item if $\fir(j+1)=j+1$, $\bM_{u_{q,j},u_{q,j+1}}=\frac{1}{2d_{u_{q,j}}^{1/2}d_{u_{q,j+1}}^{1/2}}$;
			\item if $\fir(j+1)\neq j+1$ or $\bot$, $\bM_{u_{q,j},u_{q,j+1}}=\frac{1}{2 d_{u_{q,j}}^{1/2}d_{u_{q,j+1}}^{1/2}}\leq \frac{1}{2\dmin}\cdot \frac{d_{u_{q,j}}^{1/2}}{d_{u_{q,j+1}}^{1/2}}$.
		\end{itemize}
		We have
		\begin{align*}
			&\quad\sum_{u_{q,1},\ldots,u_{q,t}:\fir(\cdot),j^*}\frac{p^{\left|U_q\setminus \left(U_1\cup\cdots\cup U_{q-1}\right)\right|}}{p^{\left|\left\{j: u_{q,j}\neq u_{q,j+1}\right\}\right|}}\cdot \prod_{j=1}^t\bM_{u_{q,j},u_{q,j+1}} \\
			&\leq \frac{1}{2^t}\sum_{u_{q,1},\ldots,u_{q,t}:\fir(\cdot),j^*}\frac{p^{\left|U_q\setminus \left(U_1\cup\cdots\cup U_{q-1}\right)\right|}}{p^{\left|\left\{j: u_{q,j}\neq u_{q,j+1}\right\}\right|}\cdot {\dmin}^{\left|\left\{j:\fir(j)\neq j,\bot\right\}\right|}}\cdot \prod_{j:\fir(j+1)=j+1}\frac{1}{d_{u_{q,j}}} \\
			&=\frac{1}{2^t}\sum_{u_{q,1},\ldots,u_{q,t}:\fir(\cdot),j^*}\frac{1}{{(p\cdot \dmin)}^{\left|\left\{j:\fir(j)\neq j,\bot\right\}\right|}}\cdot \prod_{j:\fir(j+1)=j+1}\frac{1}{d_{u_{q,j}}}.
		\end{align*}
		Now we take the sum in the order of $u_{q,j^*-1},u_{q,j^*-2},\ldots,u_{q,1},u_{q,t},\ldots,u_{q,j^*}$.
		Observe that if $\fir(j+1)\neq j+1$, then the value $u_{q,j+1}$ is \emph{determined} given all $u_{q,<j}$ and $u_{q,\geq j^*}$ that we have not taken the sum yet, i.e., there is only one term in the sum. 
		Otherwise, if $\fir(j+1)=j+1$, then we can take the sum over all possible $u_{q,j+1}$ that is a \emph{neighbor} of $u_{q,j}$ (note that we must have $j+1\neq j^*$).
		Given $u_{q,j}$, there are $d_{u_{q,j}}$ possibilities for the sum, which cancels the $\frac{1}{d_{u_{q,j}}}$ factor.
		Thus, the sum over all $U_q$ given $\fir(\cdot),j^*$ is at most
		\[
			\frac{1}{2^t}\cdot \frac{1}{{(p\cdot \dmin)}^{\left|\left\{j:\fir(j)\neq j,\bot\right\}\right|}}.
		\]
		Finally, we take the sum over all possible $\fir(\cdot)$ and $j^*$, the LHS of~\eqref{eqn_one_cycle} over all $U_q$ in this case is at most
		\begin{align*}
			\sum_{j^*=1}^t \sum_{\fir(\cdot)} \frac{1}{2^t}\cdot \frac{1}{{(p\cdot \dmin)}^{\left|\left\{j:\fir(j)\neq j,\bot\right\}\right|}}&\leq t\cdot \sum_{l=1}^t \frac{1}{2^t}\cdot \frac{1}{(p\cdot \dmin)^l}\cdot 2^{t-l}\cdot \binom{t}{l}\cdot (qt)^l \\
			&\leq t\cdot \sum_{l\geq 1}\left(\frac{qt^2}{2p\cdot \dmin}\right)^{l} \\
			&\leq \frac{qt^3}{p\cdot \dmin}.
		\end{align*}

		\paragraph{Case 3.} Finally, we consider all $U_q$ such that
		\begin{itemize}
			\item $U_q\cap (U_1\cup\cdots\cup U_{q-1})=\emptyset$;
			\item some edges $(u_{q,j-1},u_{q,j})$ with $u_{q,j-1}\neq u_{q,j}$ have the same $u_{q,j}$.
		\end{itemize}
		Consider any such $U_q$, let $w$ be the number of edges $(u_{q,j-1},u_{q,j})$ with $u_{q,j-1}\neq u_{q,j}$ such that $u_{q,j}$ appears more than once among such edges.
		Then among all $w$ such indices, there must exist one adjacent pair that is at least $\lceil t/w
		\rceil$ apart (in the cyclic order).
		Let the lexicographically first such pair be $(u_{q,j^*-1},u_{q,j^*})$.
		Thus, we have $u_{q,j^*-1}\neq u_{q,j^*}$ and for all $j=j^*-\lceil t/w\rceil+1,\ldots,j^*-1$, either $u_{q,j-1}=u_{q,j}$ or $u_{q,j}$ only appears once among all non-self-loop edges.
		Similar to Case 2, we define $\fir(\cdot)$ as follows, \emph{but only for} $j\in [t]\setminus \{j^*-\lceil t/w\rceil+1,\ldots,j^*-1\}$,
		\begin{itemize}
			\item if $u_{q,j-1}\neq u_{q,j}$, let $\fir(j)$ be first index that is at least $j^*$ in the cyclic order such that $u_{q,j}=u_{q,\fir(j)}$;
			\item if $u_{q,j-1}=u_{q,j}$, let $\fir(j)=\bot$.
		\end{itemize}

		Similarly, we will first take the sum over all $U_q$ with the same $\fir(\cdot)$, $j^*$ and $w$, then take the sum over $\fir(\cdot), j^*, w$.
		We have
		\begin{align*}
			&\, \sum_{u_{q,1},\ldots,u_{q,t}:\fir(\cdot),j^*,w}\frac{p^{\left|U_q\setminus(U_1\cup\cdots\cup U_{q-1})\right|}}{p^{\left|\left\{j: u_{q,j}\neq u_{q,j+1}\right\}\right|}}\cdot \prod_{j=1}^t\bM_{u_{q,j},u_{q,j+1}} \\
			&=\sum_{u_{q,1},\ldots,u_{q,t}:\fir(\cdot),j^*,w}\frac{1}{p^{\left|\left\{j: \fir(j)\neq j,\bot\right\}\right|}}\cdot \prod_{j=1}^t\bM_{u_{q,j},u_{q,j+1}} \\
			&\leq \sum_{\substack{u_{q,j^*},\ldots,u_{q,t},\\ u_{q,1},\ldots,u_{q,j^*-\lceil t/w\rceil}:\\ \fir(\cdot),j^*,w}}\frac{1}{p^{\left|\left\{j: \fir(j)\neq j,\bot\right\}\right|}}\cdot \left(\prod_{j\in [t]\setminus \{j^*-\lceil t/w\rceil,\ldots,j^*-1\}}\bM_{u_{q,j},u_{q,j+1}}\right)\cdot \left(\bM^{\lceil t/w\rceil}\right)_{u_{q,j^*-\lceil t/w\rceil},u_{q,j^*}} \\
			&\leq \sum_{\substack{u_{q,j^*},\ldots,u_{q,t},\\ u_{q,1},\ldots,u_{q,j^*-\lceil t/w\rceil}:\\ \fir(\cdot),j^*,w}}\frac{1}{p^{\left|\left\{j: \fir(j)\neq j,\bot\right\}\right|}}\cdot \left(\prod_{j\in [t]\setminus \{j^*-\lceil t/w\rceil,\ldots,j^*-1\}}\bM_{u_{q,j},u_{q,j+1}}\right) \\
			&\qquad\qquad\cdot d_{u_{q,j^*-\lceil t/w\rceil}}^{1/2}d_{u_{q,j^*}}^{1/2}\cdot \left(\frac{1}{D}+\left(\frac{1}{\dmin}-\frac{1}{D}\right)\cdot e^{-\varepsilon^2 t/4w}\right),
		\end{align*}
		where the last inequality is by Lemma~\ref{lem_power_u_v}.

		Similar to Case 2, we use the following bounds for each factor $\bM_{u_{q,j},u_{q,j+1}}$.
		\begin{itemize}
			\item if $\fir(j+1)=\bot$, $\bM_{u_{q,j},u_{q,j+1}}=\frac{1}{2}=\frac{1}{2}\cdot \frac{d_{u_{q,j}}^{1/2}}{d_{u_{q,j+1}}^{1/2}}$;
			\item if $\fir(j+1)=j+1$, $\bM_{u_{q,j},u_{q,j+1}}=\frac{1}{2d_{u_{q,j}}^{1/2}d_{u_{q,j+1}}^{1/2}}$;
			\item if $\fir(j+1)\neq j+1$ or $\bot$, $\bM_{u_{q,j},u_{q,j+1}}=\frac{1}{2 d_{u_{q,j}}^{1/2}d_{u_{q,j+1}}^{1/2}}\leq \frac{1}{2\dmin}\cdot \frac{d_{u_{q,j}}^{1/2}}{d_{u_{q,j+1}}^{1/2}}$.
		\end{itemize}
		The sum is at most
		\begin{align*}
			&\, \frac{1}{2^{t-\lceil t/w\rceil}}\sum_{\substack{u_{q,j^*},\ldots,u_{q,t},\\ u_{q,1},\ldots,u_{q,j^*-\lceil t/w\rceil}:\\ \fir(\cdot),j^*,w}}\frac{1}{(p\cdot\dmin)^{\left|\left\{j: \fir(j)\neq j,\bot\right\}\right|}}\cdot \left(\prod_{j\in [t]\setminus \{j^*-\lceil t/w\rceil,\ldots,j^*-1\}:\fir(j+1)=j+1}\frac{1}{d_{u_{q,j}}}\right) \\
			&\qquad\qquad\cdot d_{u_{q,j^*}}\cdot \left(\frac{1}{D}+\left(\frac{1}{\dmin}-\frac{1}{D}\right)\cdot e^{-\varepsilon^2 t/4w}\right),
		\end{align*}
		which by taking the sum in the order of $u_{q,j^*-\lceil t/w\rceil},u_{q,j^*-\lceil t/w\rceil-1},\ldots,u_{q,1},u_{q,t},\ldots,u_{q,j^*}$, is at most
		\begin{align*}
			&\, \frac{1}{2^{t-\lceil t/w\rceil}}\sum_{u_{q,j^*}} \frac{1}{(p\cdot\dmin)^{\left|\left\{j: \fir(j)\neq j,\bot\right\}\right|}}\cdot d_{u_{q,j^*}}\cdot \left(\frac{1}{D}+\left(\frac{1}{\dmin}-\frac{1}{D}\right)\cdot e^{-\varepsilon^2 t/4w}\right) \\
			&=\frac{1}{2^{t-\lceil t/w\rceil}}\cdot\frac{1}{(p\cdot\dmin)^{\left|\left\{j: \fir(j)\neq j,\bot\right\}\right|}}\cdot \left(1+\left(\frac{D}{\dmin}-1\right)\cdot e^{-\varepsilon^2 t/4w}\right).
		\end{align*}

		Finally, observe that we must have $w/2\leq \left|\left\{j: \fir(j)\neq j,\bot\right\}\right|\leq w$.
		By summing over all $\fir(\cdot),j^*,w$, the LHS of~\eqref{eqn_one_cycle} over all $U_q$ in this case is at most
		\begin{align*}
			&\, \sum_{w=2}^t\sum_{j^*=1}^t\sum_{\fir(\cdot)}\frac{1}{2^{t-\lceil t/w\rceil}}\cdot\frac{1}{(p\cdot\dmin)^{\left|\left\{j: \fir(j)\neq j,\bot\right\}\right|}}\cdot \left(1+\left(\frac{D}{\dmin}-1\right)\cdot e^{-\varepsilon^2 t/4w}\right) \\
			&=\sum_{w=2}^t\sum_{j^*=1}^t\sum_{l=w/2}^{w}\frac{1}{2^{t-\lceil t/w\rceil}}\cdot\frac{1}{(p\cdot\dmin)^l}\cdot \left(1+\left(\frac{D}{\dmin}-1\right)\cdot e^{-\varepsilon^2 t/4w}\right)\cdot 2^{t-\lceil t/w\rceil-l}\cdot \binom{t}{l}\cdot t^l \\
			&\leq \sum_{w=2}^t t\cdot \sum_{l=w/2}^{w}\left(\frac{t^2}{2p\cdot\dmin}\right)^l\cdot \left(1+\left(\frac{D}{\dmin}-1\right)\cdot e^{-\varepsilon^2 t/4w}\right) \\
			&\leq 2t\cdot \left(\sum_{w=2}^{t/(16\varepsilon^{-2}\ln n)}+\sum_{w\geq t/(16\varepsilon^{-2}\ln n)}\right) \left(\frac{t^2}{2p\cdot\dmin}\right)^{w/2}\cdot \left(1+\left(\frac{D}{\dmin}-1\right)\cdot e^{-\varepsilon^2 t/4w}\right) \\
			&\leq \frac{2t^3}{p\cdot \dmin}\cdot \left(1+n^{-2}\right)+4t\cdot \left(\frac{t^2}{2p\cdot \dmin}\right)^{t/(32\varepsilon^{-2}\ln n)}\cdot D \\
			&\leq \frac{2t^3}{p\cdot \dmin}\cdot \left(1+n^{-2}\right)+4tn^2\cdot \left(\frac{t^2}{2p\cdot \dmin}\right)^{8/\delta} \\
			&\leq \frac{2t^3}{p\cdot \dmin}\cdot \left(1+n^{-2}\right)+4tn^{-2} \\
			&\leq \frac{4t^3}{p\cdot \dmin}.
		\end{align*}

		Summing up all three cases proves the lemma.
	\end{proof}

%% file: implement.tex

\section{Implementing Prior Work Via Random Gaussian Sketches}\label{app:implement}

We now outline implementations of existing works on graph sketching in our model.  

\subsection{$\ell_0$-samplers and connectivity sketches}\label{sec:l0-samplers}

Recall that in the $\ell_0$-sampling problem one needs to design a sketching matrix $A$ such that for every $x\in \R^n$ one can recover a uniformly random element of $x$ from $Ax$ (to within total variation distance $\delta$) or output FAIL (with failure probability bounded by $\delta$)\footnote{Note that these two parameters appear differently in the space complexity of $\ell_0$-sampling, and are therefore treated separately in works that obtain optimal space bounds for $\ell_0$-samplers~\cite{AhnGM12a}. We set both parameters to $\delta$ for simplicity.}. We outline a construction of an $\ell_0$-sampler in our model, i.e. where every rows of the sketch $A$ is of the form $g\cdot S$, where $S$ is an arbitrary matrix with zeros and ones on the diagonal and zeros on off-diagonal entries ($S$ is known to the decoder) and $g$ is a vector with i.i.d. unit variance Gaussian entries ($g$ is not known to the decoder).  Note that our $\ell_0$-sampler only needs to work for vectors $x$ whose entries are in $\{-1, 0, +1\}$, as this is the case in all applications of graph sketching.

We first recall the construction of a basic $\ell_0$ sampler (see~\cite{JowhariST11} for a space-optimal construction). For integer $j$ between $0$ and $\lceil \log_2 n\rceil$ let $x^j\in \R^n$ denote the restriction of $x^j$ to elements of a subset of the universe $[n]$ that includes every element independently with probability $2^{-j}$.  There exists $j^*$ such that with constant probability $x^j$ contains exactly one nonzero. To determine the value of $j^*$ or conclude that such an index does not exist, it suffices to estimate the $\ell_2^2$ norm of $x$ to within a $1\pm 1/3$ factor, for example (since nonzero entries of $x$ equal $1$ in absolute value). The latter can be achieved (with at most inverse polynomial failure probability) by averaging squared dot products of $O(\log n)$ independent Gaussian vectors with $x^j$, which is allowed by our model. Note that here the decoder indeed does not need to know the Gaussian vectors, as required.  If $j^*$ exists, one must recover the identity of the nonzero element. The typical way to do it is to compute the dot product of $x^{j^*}$ with the vector whose $i$-th coordinate equals $i$, for every $i\in [n]$. This is not available in our model. To replace this approach, for every $j=0,\ldots, \lceil \log_2 n\rceil$ and $b=0,\ldots, \lceil \log_2 n\rceil$ approximate the $\ell_2^2$ norm of the vector $x^j$ restricted to the set of elements in $[n]$ that have $1$ in the $b$-th position in their binary representation using $O(\log n)$ dot products with i.i.d. Gaussians. This allows one to read off the binary representation of the nonzero in $x^{j^*}$, and therefore yields an $\ell_0$ sampler.

\paragraph{Graph connectivity and spanning trees.} Since an $\ell_0$-sampler is the only sketch used by the connectivity sketch of~\cite{AhnGM12a}, it follows that a spanning forest of the input graph can be recovered by a sketch that fits our model and has a polylogarithmic number of rows. 

\paragraph{Approximate vertex connectivity.} The result of~\cite{GuhaMT15} uses the spanning tree sketch of~\cite{AhnGM12a} black box (the sketch is applied to random vertex induced subgraphs) to approximate vertex connectivity. Since the sketch of~\cite{AhnGM12a} can be implemented in our model, as described above, the result of~\cite{GuhaMT15} also can.

\subsection{$\ell_2$-heavy hitters, spectral sparsifiers and spanners}
Recall that in the $\varphi$-heavy hitters problem in $\ell_2$ one needs to design a sketching matrix $A$ such that for every $x\in \R^n$ one can recover a list of elements $L\subseteq [n]$ such that every $i\in [n]$ satisfying $x_i^2\geq \varphi \|x\|_2^2$ belongs to $L$ and no $i\in [n]$ with $x_i^2<c\varphi \|x\|_2^2$ for a constant $c>0$ belongs to $L$.

A basic $\ell_2$ heavy hitters sketch works by first hashing elements of $[n]$ to $B\approx 1/\varphi$ buckets, i.e. effectively defining $x^b$ for $b\in [B]$ to be the restriction of $x$ to bucket $b$, and computing the sum of elements of $x^b$ with random signs. In our model we can replace the random signs with random Gaussians, so that the resulting dot product is Gaussian with variance $\|x^b\|_2^2$. Fixing any $j\in [n]$ and letting $b$ denote the bucket that $j$ hashes to we get that a single hashing  can be used to obtain an estimate of its absolute value that is correct up to constant factor and an additive $O(1/\sqrt{B}) \|x\|_2$ term with probability\footnote{We use the fact that the dot product of $x^b$ with a random Gaussian vector will be distributed as $g_j x_j+N(0, \|x_{-j}^b\|_2^2)$, and $|g_j|$ is at least a constant with probability at least $9/10$. Here $x^b_j$ stands for the vector obtained from $x^b$ by zeroing out entry $j$.} at least $9/10$. We can now repeat the estimator $O(\log n)$ times and include in $L$ elements that are estimated as larger than a $c' \varphi \|x\|_2$ for a sufficiently small constant $c'>0$ in absolute value. Therefore, setting $B=O(1/\varphi)$ achieves the required bounds. This yields an $\ell_2$-heavy hitters sketch with decoding time nearly linear in the size $n$ of the universe. The decoding time can be improved to $(1/\varphi)\cdot \text{poly}(\log n)$ using a bit-encoding approach similar to the one from Section~\ref{sec:l0-samplers} above.

\paragraph{Spectral sparsifiers and spanners.} Spectral sparsification sketches~\cite{KapralovLMMS14,KNST19,KapralovMMMNST20} require graph connectivity sketches, which we already implemented in Section~\ref{sec:l0-samplers}, as well as $\ell_2$-heavy hitters sketches, and therefore can also be implemented in our model. Non-adaptive sketching algorithms for spanner construction~\cite{FiltserKN21} rely on spectral sparsification sketches that are applied to vertex-induced subgraphs of the input graph. Thus, these sketches can also be implemented in our model with at most a polylogarithmic loss in the number of rows.

%% file: background.tex

\section{Background}\label{sec:background} 

This appendix includes a summary of basic tools from probability and information theory, and spectral graph theory that we use in our paper.  

\subsection{Background in Probability and Information Theory}\label{sec:info} 

The proof of basic facts included in this part can be found in~\cite{CoverT06}. 

\paragraph{KL-divergence.} For continuous distributions $P$ and $Q$, the \emph{Kullback--Leibler divergence} (KL-divergence) of $P$ from $Q$ is 
\[
	\KLD{P}{Q}{}=\expect_{X\sim P}\left[\log \left(\frac{\mathbf{d}P(X)}{\mathbf{d}Q(X)}\right)\right].
\]
We may abuse the notation $\KLD{X}{Y}{}$ for random variables $X$ and $Y$ to denote the KL-divergence of the distribution of $X$ from the distribution of $Y$.

\begin{fact}[Chain rule of KL-divergence]\label{fact:kl-chain-rule}
	Let $P(X,Y)$ and $Q(X,Y)$ be two distributions for a pair of random variables $X$ and $Y$. Then, 
	\[
		\KLD{P(X,Y)}{Q(X,Y)}{} = \KLD{P(X)}{Q(X)}{} + \Exp_{x \sim P} \bracket{\KLD{P(Y \mid X=x)}{Q(Y \mid X=x)}{}}. 
	\]
	In particular, if $X \perp Y$ in both distributions $P$ and $Q$, then, 
	\[
		\KLD{P(X,Y)}{Q(X,Y)}{} =  \KLD{P(X)}{Q(X)}{} +  \KLD{P(Y)}{Q(Y)}{}. 
	\]
\end{fact}

\paragraph{Total variation distance.} Similarly, for continuous distributions $P$ and $Q$ over the same sample space $\Omega$, the \emph{total variation distance} (TVD) between $P$ and $Q$ is
\[
	\tvd{P}{Q} = \sup_{\Omega' \subseteq \Omega}\card{P(\Omega') - Q(\Omega')}. 
\]

\begin{fact}\label{fact:tvd-sample}
	Suppose we are given a single sample $s$ chosen uniformly at random from either distribution $P$ or $Q$. The best probability of success in determining the source of $s$ is
	\[
		\frac{1}{2} + \frac{1}{2} \cdot \tvd{P}{Q}. 
	\]
\end{fact}

\begin{fact}\label{fact:tvd-chain-rule}
	Let $P(X,Y)$ and $Q(X,Y)$ be two distributions for a pair of random variables $X$ and $Y$. Then, 
	\[
		\tvd{P(X,Y)}{Q(X,Y)} \leq \tvd{P(X)}{Q(X)} + \Exp_{x \sim P} \bracket{\tvd{P(Y \mid X=x)}{Q(Y \mid X=x)}}. 
	\]
	In particular, if $X$ has the same marginal distribution in $P$ and $Q$, then
	\[
		\tvd{P(X,Y)}{Q(X,Y)} \leq  \Exp_{x \sim P}\tvd{P(Y \mid X=x)}{Q(Y \mid X=x)}. 
	\]
\end{fact}

Pinsker's inequality relates KL-divergence to total variation distance. 

\begin{fact}[Pinsker's inequality]\label{fact:pinsker}
	For any pairs of distributions $P$ and $Q$ over the same domain, 
	\[
		\tvd{P}{Q} \leq \sqrt{\frac{1}{2} \cdot \KLD{P}{Q}{}}. 
	\]
\end{fact}

\subsection{Background in Spectral Graph Theory}\label{sec:spectral}

	Let $G=(V,E)$ be an undirected graph on $n$ vertices.
	Let $d_1,\ldots,d_n$ be the degrees. 
	For a vertex set $S\subseteq V$, its \emph{volume} in $G$ is 
	\[
	\vol_G(S) := \sum_{u\in S}d_u.
	\]
	When there is no ambiguity, we may omit the subscript $G$, and denote it by $\vol(S)$.

	For $S,T\subseteq V$, $E(S,T)$ is the set of edges between $S$ and $T$, i.e., $E(S,T):=\{(u,v)\in E:u\in S,v\in T\}$.
	The \emph{conductance} of $G$ is 
	\[
	\varphi(G) := \min_{S\subseteq V}\frac{\left|E(S,V\setminus S)\right|}{\min\{\vol(S), \vol(V\setminus S)\}}.
	\]
	We say that $G$ is a \emph{$\varphi$-expander} if its conductance is at least $\varphi$.

	We associate the following matrices with a graph $G=(V,E)$: 
	\begin{itemize}
	\item The \emph{adjacency matrix} $\bA$ of $G$ the $n \times n$ matrix such that $\bA_{u,v} = 1$ iff $(u,v)$ is an edge in $E$. 
	\item The \emph{degree diagonal matrix} $\bD$ of $G$ is the $n\times n$ matrix such that $\bD_{v,v}=d_v$ and $0$ elsewhere.
	\item The \emph{signed edge-incidence matrix} $\bB$ of $G$ is the ${{n}\choose{2}} \times n$ matrix such that $\bB_{(u,v),w}$ is $1$ if $w=u$ and $(u,v)$ is an edge in $E$, $-1$ if $w=v$ and $(u,v)$ is an edge in $E$, and $0$ otherwise. 
	\item The \emph{Laplacian matrix} $\bL$ is the $n\times n$ matrix 
	\[
	\bI-\bA=\sum_{e=(u,v)\in E}b_eb_e^\top,
	\]
	where $b_e$ is the vector with value $1$ in coordinate $u$, $-1$ in coordinate $v$, and $0$ in all other coordinates.
	\item The \emph{normalized Laplacian matrix} $\widetilde{\mathbf{L}}$ is the $n\times n$ matrix such that
	\[
		\widetilde{\mathbf{L}}_{i,j}=\begin{cases}
			1 & i=j, \\
			-\frac{1}{(d_id_j)^{1/2}} & (i,j)\in E, \\
			0 & \textrm{o.w.}
		\end{cases}
	\]
	Equivalently, 
	\[
	\widetilde{\mathbf{L}}=\bI-\bD^{-1/2}\bA\bD^{-1/2}.
	\]
	\end{itemize}
\paragraph{Spectral gap:} 	For any graph $G$, both $\bL$ and $\widetilde{\mathbf{L}}$ are positive semidefinite.
	For $\widetilde{\mathbf{L}}$, the smallest eigenvalue is zero, with eigenvector $\bD^{1/2}\bOne$, where $\bOne$ is the all-one vector.
	Its largest eigenvalue is at most $2$.
	The \emph{spectral gap} of $G$ is the second smallest eigenvalue of $\widetilde{\mathbf{L}}$.
	Cheeger's inequality relates the conductance and the spectral gap.

	\begin{proposition}[Cheeger's inequality,~\cite{Chung96} Section 3]\label{lem:conduct_spgap}
		For any graph $G$, let $\lambda$ be its spectral gap, and $\varphi$ be its conductance, we have
		\[
			2\varphi\geq \lambda\geq \frac{\varphi^2}{2}.
		\]
	\end{proposition}

\paragraph{Effective resistance.} By treating each edge of $G$ as a resistor with unit resistance, we denote the \emph{effective resistance} between $u, v$ by $\eff{G}{u,v}$.
	For any pair of vertices $(u, v)$, we have
	\[
		\eff{G}{u,v}=b_{(u, v)}^{\top}\bL^{+}b_{(u, v)},
	\]
	where $\bL^+$ is the \emph{pseudoinverse} of $\bL$.

\paragraph{Expander decomposition.} We also use the following (variant of) expander decomposition that bounds the minimum degree of resulting expanders. The proof is a simple modification of standard decompositions, e.g. in~\cite{VempalaV00,SaranurakW19}, and is provided only for completeness.

\begin{proposition*}[Restatement of~\Cref{prop:expander-decomposition}]
	Let $G=(V, E)$ be any arbitrary graph on $n$ vertices and $m$ edges, and $\eps \in (0,1/2)$ and $\dmin \geq 1$ be parameters. The vertices of $G$ can be partitioned into subgraphs $H_1,\ldots,H_k$ such that: 
	\begin{enumerate}[label=($\roman*)$]
		\item Each $H_i$ is an $\eps$-expander with minimum degree $\dmin$; 
		\item At most $8\eps \cdot m\log{n} + n \cdot \dmin$ edges $E_0$ of $G$ do not belong to any subgraph $\set{H_i}_{i \in [k]}$. 
	\end{enumerate}
\end{proposition*}
\begin{proof}
	We use a standard sparse cut pruning (contributing the first term in extra edges) plus low degree pruning for vertices with degree $< \dmin$ (contributing the second term). 	
	
	Firstly, remove all \emph{vertices} with degree $<\dmin$ from $G$ and include their edges in $E_0$.
	If $G$ at this point is an $\eps$-expander, we terminate.
	Otherwise find a set $S\subseteq V$ such that $\vol(S)\leq \frac{1}{2}\vol(V)$ and $\left|E(S,V\setminus S)\right| < \eps \cdot \vol(S)$. 
	Partition the graph along the cut $(S, V\setminus S)$ and insert all edges in $E(S,V\setminus S)$ to $E_0$. 
	Recursively repeat both steps on the subgraphs induced by $S$ and $V\setminus S$ respectively. At the end, return the resulting expanders as the collection $H_1,\ldots,H_k$ in 
	part $(i)$ and the removed edges $E_0$ as the extra edges in part $(ii)$. 
	
	Firstly, we will only include at most $n \cdot \dmin$ edges in $E_0$ when removing vertices with degree $<\dmin$ in total. Moreover, whenever we remove edges $E(S,V\setminus S)$, we can ``charge'' 
	the removed edges to the edges with both endpoints in $S$.
	Since $\left|E(S,V\setminus S)\right| < \eps \cdot \vol(S)$, every edge in $S$ is charged at most $2\varepsilon/(1-\varepsilon)\leq 4\varepsilon$.
	Since $\vol(S)\leq \frac{1}{2}\vol(V)$, which is the smaller side of the cut, each edge can be charged at most $\log\!{{n}\choose{2}}<2\log n$ times, summing up to at most $8\varepsilon\log n$.
	This proves the proposition.
\end{proof}

%% file: main.bbl
\begin{thebibliography}{KMM{\etalchar{+}}20}

\bibitem[ACK19]{AssadiCK19}
Sepehr Assadi, Yu~Chen, and Sanjeev Khanna.
\newblock Sublinear algorithms for ({\(\Delta\)} + 1) vertex coloring.
\newblock In Timothy~M. Chan, editor, {\em Proceedings of the Thirtieth Annual
  {ACM-SIAM} Symposium on Discrete Algorithms, {SODA} 2019, San Diego,
  California, USA, January 6-9, 2019}, pages 767--786. {SIAM}, 2019.

\bibitem[AGM12a]{AhnGM12a}
Kook~Jin Ahn, Sudipto Guha, and Andrew McGregor.
\newblock Analyzing graph structure via linear measurements.
\newblock In {\em Proceedings of the Twenty-Third Annual {ACM-SIAM} Symposium
  on Discrete Algorithms, {SODA} 2012, Kyoto, Japan, January 17-19, 2012},
  pages 459--467, 2012.

\bibitem[AGM12b]{AhnGM12b}
Kook~Jin Ahn, Sudipto Guha, and Andrew McGregor.
\newblock Graph sketches: sparsification, spanners, and subgraphs.
\newblock In {\em Proceedings of the 31st {ACM} {SIGMOD-SIGACT-SIGART}
  Symposium on Principles of Database Systems, {PODS} 2012, Scottsdale, AZ,
  USA, May 20-24, 2012}, pages 5--14, 2012.

\bibitem[AHLW16]{AiHLW16}
Yuqing Ai, Wei Hu, Yi~Li, and David~P. Woodruff.
\newblock New characterizations in turnstile streams with applications.
\newblock In Ran Raz, editor, {\em 31st Conference on Computational Complexity,
  {CCC} 2016, May 29 to June 1, 2016, Tokyo, Japan}, volume~50 of {\em LIPIcs},
  pages 20:1--20:22. Schloss Dagstuhl - Leibniz-Zentrum f{\"{u}}r Informatik,
  2016.

\bibitem[AKLY16]{AssadiKLY16}
Sepehr Assadi, Sanjeev Khanna, Yang Li, and Grigory Yaroslavtsev.
\newblock Maximum matchings in dynamic graph streams and the simultaneous
  communication model.
\newblock In {\em Proceedings of the Twenty-Seventh Annual {ACM-SIAM} Symposium
  on Discrete Algorithms, {SODA} 2016, Arlington, VA, USA, January 10-12,
  2016}, pages 1345--1364, 2016.

\bibitem[AKO20]{AssadiKO20}
Sepehr Assadi, Gillat Kol, and Rotem Oshman.
\newblock Lower bounds for distributed sketching of maximal matchings and
  maximal independent sets.
\newblock In Yuval Emek and Christian Cachin, editors, {\em {PODC} '20: {ACM}
  Symposium on Principles of Distributed Computing, Virtual Event, Italy,
  August 3-7, 2020}, pages 79--88. {ACM}, 2020.

\bibitem[AMS96]{AlonMS96}
Noga Alon, Yossi Matias, and Mario Szegedy.
\newblock The space complexity of approximating the frequency moments.
\newblock In Gary~L. Miller, editor, {\em Proceedings of the Twenty-Eighth
  Annual {ACM} Symposium on the Theory of Computing, Philadelphia,
  Pennsylvania, USA, May 22-24, 1996}, pages 20--29. {ACM}, 1996.

\bibitem[AS22]{AssadiS22}
Sepehr Assadi and Vihan Shah.
\newblock An asymptotically optimal algorithm for maximum matching in dynamic
  streams.
\newblock In Mark Braverman, editor, {\em 13th Innovations in Theoretical
  Computer Science Conference, {ITCS} 2022, January 31 - February 3, 2022,
  Berkeley, CA, {USA}}, volume 215 of {\em LIPIcs}, pages 9:1--9:23. Schloss
  Dagstuhl - Leibniz-Zentrum f{\"{u}}r Informatik, 2022.

\bibitem[BMN{\etalchar{+}}11]{BeckerMNRST11}
Florent Becker, Mart{\'{\i}}n Matamala, Nicolas Nisse, Ivan Rapaport, Karol
  Suchan, and Ioan Todinca.
\newblock Adding a referee to an interconnection network: What can(not) be
  computed in one round.
\newblock In {\em 25th {IEEE} International Symposium on Parallel and
  Distributed Processing, {IPDPS} 2011, Anchorage, Alaska, USA, 16-20 May, 2011
  - Conference Proceedings}, pages 508--514. {IEEE}, 2011.

\bibitem[BMRT14]{BeckerMRT14}
Florent Becker, Pedro Montealegre, Ivan Rapaport, and Ioan Todinca.
\newblock The simultaneous number-in-hand communication model for networks:
  Private coins, public coins and determinism.
\newblock In Magn{\'{u}}s~M. Halld{\'{o}}rsson, editor, {\em Structural
  Information and Communication Complexity - 21st International Colloquium,
  {SIROCCO} 2014, Takayama, Japan, July 23-25, 2014. Proceedings}, volume 8576
  of {\em Lecture Notes in Computer Science}, pages 83--95. Springer, 2014.

\bibitem[BS07]{BaswanaS07}
Surender Baswana and Sandeep Sen.
\newblock A simple and linear time randomized algorithm for computing sparse
  spanners in weighted graphs.
\newblock {\em Random Struct. Algorithms}, 30(4):532--563, 2007.

\bibitem[CCF02]{CharikarCF02}
Moses Charikar, Kevin~C. Chen, and Martin Farach{-}Colton.
\newblock Finding frequent items in data streams.
\newblock In Peter Widmayer, Francisco~Triguero Ruiz, Rafael~Morales Bueno,
  Matthew Hennessy, Stephan~J. Eidenbenz, and Ricardo Conejo, editors, {\em
  Automata, Languages and Programming, 29th International Colloquium, {ICALP}
  2002, Malaga, Spain, July 8-13, 2002, Proceedings}, volume 2380 of {\em
  Lecture Notes in Computer Science}, pages 693--703. Springer, 2002.

\bibitem[Chu96]{Chung96}
Fan~RK Chung.
\newblock Laplacians of graphs and cheeger’s inequalities.
\newblock {\em Combinatorics, Paul Erdos is Eighty}, 2(157-172):13--2, 1996.

\bibitem[CKL]{ChenKL}
Yu~Chen, Sanjeev Khanna, and Huan Li.
\newblock On weighted graph sparsification by linear sketching.
\newblock In {\em FOCS 22}.

\bibitem[CM04]{CormodeM04}
Graham Cormode and S.~Muthukrishnan.
\newblock An improved data stream summary: The count-min sketch and its
  applications.
\newblock In Martin Farach{-}Colton, editor, {\em {LATIN} 2004: Theoretical
  Informatics, 6th Latin American Symposium, Buenos Aires, Argentina, April
  5-8, 2004, Proceedings}, volume 2976 of {\em Lecture Notes in Computer
  Science}, pages 29--38. Springer, 2004.

\bibitem[CT06]{CoverT06}
Thomas~M. Cover and Joy~A. Thomas.
\newblock {\em Elements of information theory {(2.} ed.)}.
\newblock Wiley, 2006.

\bibitem[DK20]{DarkK20}
Jacques Dark and Christian Konrad.
\newblock Optimal lower bounds for matching and vertex cover in dynamic graph
  streams.
\newblock In Shubhangi Saraf, editor, {\em 35th Computational Complexity
  Conference, {CCC} 2020, July 28-31, 2020, Saarbr{\"{u}}cken, Germany (Virtual
  Conference)}, volume 169 of {\em LIPIcs}, pages 30:1--30:14. Schloss Dagstuhl
  - Leibniz-Zentrum f{\"{u}}r Informatik, 2020.

\bibitem[Don06]{Donoho06}
David~L Donoho.
\newblock Compressed sensing.
\newblock {\em IEEE Transactions on information theory}, 52(4):1289--1306,
  2006.

\bibitem[ET21]{ElkinT21}
Michael Elkin and Chhaya Trehan.
\newblock {\textdollar}(1+{\(\epsilon\)}){\textdollar}-approximate shortest
  paths in dynamic streams.
\newblock {\em CoRR}, abs/2107.13309, 2021.

\bibitem[FKM{\etalchar{+}}04]{FeigenbaumKMSZ04}
Joan Feigenbaum, Sampath Kannan, Andrew McGregor, Siddharth Suri, and Jian
  Zhang.
\newblock On graph problems in a semi-streaming model.
\newblock In Josep D{\'{\i}}az, Juhani Karhum{\"{a}}ki, Arto Lepist{\"{o}}, and
  Donald Sannella, editors, {\em Automata, Languages and Programming: 31st
  International Colloquium, {ICALP} 2004, Turku, Finland, July 12-16, 2004.
  Proceedings}, volume 3142 of {\em Lecture Notes in Computer Science}, pages
  531--543. Springer, 2004.

\bibitem[FKM{\etalchar{+}}08]{FeigenbaumKMSZ08}
Joan Feigenbaum, Sampath Kannan, Andrew McGregor, Siddharth Suri, and Jian
  Zhang.
\newblock Graph distances in the data-stream model.
\newblock {\em {SIAM} J. Comput.}, 38(5):1709--1727, 2008.

\bibitem[FKN21]{FiltserKN21}
Arnold Filtser, Michael Kapralov, and Navid Nouri.
\newblock Graph spanners by sketching in dynamic streams and the simultaneous
  communication model.
\newblock In D{\'{a}}niel Marx, editor, {\em Proceedings of the 2021 {ACM-SIAM}
  Symposium on Discrete Algorithms, {SODA} 2021, Virtual Conference, January 10
  - 13, 2021}, pages 1894--1913. {SIAM}, 2021.

\bibitem[FWY20]{FernandezW020}
Manuel Fernandez, David~P. Woodruff, and Taisuke Yasuda.
\newblock Graph spanners in the message-passing model.
\newblock In Thomas Vidick, editor, {\em 11th Innovations in Theoretical
  Computer Science Conference, {ITCS} 2020, January 12-14, 2020, Seattle,
  Washington, {USA}}, volume 151 of {\em LIPIcs}, pages 77:1--77:18. Schloss
  Dagstuhl - Leibniz-Zentrum f{\"{u}}r Informatik, 2020.

\bibitem[GMT15]{GuhaMT15}
Sudipto Guha, Andrew McGregor, and David Tench.
\newblock Vertex and hyperedge connectivity in dynamic graph streams.
\newblock In Tova Milo and Diego Calvanese, editors, {\em Proceedings of the
  34th {ACM} Symposium on Principles of Database Systems, {PODS} 2015,
  Melbourne, Victoria, Australia, May 31 - June 4, 2015}, pages 241--247.
  {ACM}, 2015.

\bibitem[GP16]{GhaffariP16}
Mohsen Ghaffari and Merav Parter.
\newblock {MST} in log-star rounds of congested clique.
\newblock In George Giakkoupis, editor, {\em Proceedings of the 2016 {ACM}
  Symposium on Principles of Distributed Computing, {PODC} 2016, Chicago, IL,
  USA, July 25-28, 2016}, pages 19--28. {ACM}, 2016.

\bibitem[HPP{\etalchar{+}}15]{HegemanPPSS15}
James~W. Hegeman, Gopal Pandurangan, Sriram~V. Pemmaraju, Vivek~B. Sardeshmukh,
  and Michele Scquizzato.
\newblock Toward optimal bounds in the congested clique: Graph connectivity and
  {MST}.
\newblock In Chryssis Georgiou and Paul~G. Spirakis, editors, {\em Proceedings
  of the 2015 {ACM} Symposium on Principles of Distributed Computing, {PODC}
  2015, Donostia-San Sebasti{\'{a}}n, Spain, July 21 - 23, 2015}, pages
  91--100. {ACM}, 2015.

\bibitem[JL84]{JohnsonL84}
William~B Johnson and Joram Lindenstrauss.
\newblock Extensions of lipschitz mappings into a hilbert space 26.
\newblock {\em Contemporary mathematics}, 26:28, 1984.

\bibitem[JN18]{JurdzinskiN18}
Tomasz Jurdzinski and Krzysztof Nowicki.
\newblock {MST} in \emph{O}(1) rounds of congested clique.
\newblock In Artur Czumaj, editor, {\em Proceedings of the Twenty-Ninth Annual
  {ACM-SIAM} Symposium on Discrete Algorithms, {SODA} 2018, New Orleans, LA,
  USA, January 7-10, 2018}, pages 2620--2632. {SIAM}, 2018.

\bibitem[JST11]{JowhariST11}
Hossein Jowhari, Mert Saglam, and G{\'{a}}bor Tardos.
\newblock Tight bounds for lp samplers, finding duplicates in streams, and
  related problems.
\newblock In Maurizio Lenzerini and Thomas Schwentick, editors, {\em
  Proceedings of the 30th {ACM} {SIGMOD-SIGACT-SIGART} Symposium on Principles
  of Database Systems, {PODS} 2011, June 12-16, 2011, Athens, Greece}, pages
  49--58. {ACM}, 2011.

\bibitem[KLM{\etalchar{+}}14]{KapralovLMMS14}
Michael Kapralov, Yin~Tat Lee, Cameron Musco, Christopher Musco, and Aaron
  Sidford.
\newblock Single pass spectral sparsification in dynamic streams.
\newblock In {\em 55th {IEEE} Annual Symposium on Foundations of Computer
  Science, {FOCS} 2014, Philadelphia, PA, USA, October 18-21, 2014}, pages
  561--570. {IEEE} Computer Society, 2014.

\bibitem[KMM{\etalchar{+}}20]{KapralovMMMNST20}
Michael Kapralov, Aida Mousavifar, Cameron Musco, Christopher Musco, Navid
  Nouri, Aaron Sidford, and Jakab Tardos.
\newblock Fast and space efficient spectral sparsification in dynamic streams.
\newblock In Shuchi Chawla, editor, {\em Proceedings of the 2020 {ACM-SIAM}
  Symposium on Discrete Algorithms, {SODA} 2020, Salt Lake City, UT, USA,
  January 5-8, 2020}, pages 1814--1833. {SIAM}, 2020.

\bibitem[KNP{\etalchar{+}}17]{KapralovNPWWY17}
Michael Kapralov, Jelani Nelson, Jakub Pachocki, Zhengyu Wang, David~P.
  Woodruff, and Mobin Yahyazadeh.
\newblock Optimal lower bounds for universal relation, and for samplers and
  finding duplicates in streams.
\newblock In Chris Umans, editor, {\em 58th {IEEE} Annual Symposium on
  Foundations of Computer Science, {FOCS} 2017, Berkeley, CA, USA, October
  15-17, 2017}, pages 475--486. {IEEE} Computer Society, 2017.

\bibitem[KNST19]{KNST19}
Michael Kapralov, Navid Nouri, Aaron Sidford, and Jakab Tardos.
\newblock Dynamic streaming spectral sparsification in nearly linear time and
  space.
\newblock {\em CoRR}, abs/1903.12150, 2019.

\bibitem[KP20]{KallaugherP20}
John Kallaugher and Eric Price.
\newblock Separations and equivalences between turnstile streaming and linear
  sketching.
\newblock In Konstantin Makarychev, Yury Makarychev, Madhur Tulsiani, Gautam
  Kamath, and Julia Chuzhoy, editors, {\em Proccedings of the 52nd Annual {ACM}
  {SIGACT} Symposium on Theory of Computing, {STOC} 2020, Chicago, IL, USA,
  June 22-26, 2020}, pages 1223--1236. {ACM}, 2020.

\bibitem[KVV00]{VempalaV00}
Ravi Kannan, Santosh~S. Vempala, and Adrian Vetta.
\newblock On clusterings - good, bad and spectral.
\newblock In {\em 41st Annual Symposium on Foundations of Computer Science,
  {FOCS} 2000, 12-14 November 2000, Redondo Beach, California, {USA}}, pages
  367--377. {IEEE} Computer Society, 2000.

\bibitem[KW14]{KapralovW14}
Michael Kapralov and David~P. Woodruff.
\newblock Spanners and sparsifiers in dynamic streams.
\newblock In Magn{\'{u}}s~M. Halld{\'{o}}rsson and Shlomi Dolev, editors, {\em
  {ACM} Symposium on Principles of Distributed Computing, {PODC} '14, Paris,
  France, July 15-18, 2014}, pages 272--281. {ACM}, 2014.

\bibitem[LNW14]{LiNW14}
Yi~Li, Huy~L. Nguyen, and David~P. Woodruff.
\newblock Turnstile streaming algorithms might as well be linear sketches.
\newblock In David~B. Shmoys, editor, {\em Symposium on Theory of Computing,
  {STOC} 2014, New York, NY, USA, May 31 - June 03, 2014}, pages 174--183.
  {ACM}, 2014.

\bibitem[MTVV15]{McGregorTVV15}
Andrew McGregor, David Tench, Sofya Vorotnikova, and Hoa~T. Vu.
\newblock Densest subgraph in dynamic graph streams.
\newblock In Giuseppe~F. Italiano, Giovanni Pighizzini, and Donald Sannella,
  editors, {\em Mathematical Foundations of Computer Science 2015 - 40th
  International Symposium, {MFCS} 2015, Milan, Italy, August 24-28, 2015,
  Proceedings, Part {II}}, volume 9235 of {\em Lecture Notes in Computer
  Science}, pages 472--482. Springer, 2015.

\bibitem[NY19]{NelsonY19}
Jelani Nelson and Huacheng Yu.
\newblock Optimal lower bounds for distributed and streaming spanning forest
  computation.
\newblock In Timothy~M. Chan, editor, {\em Proceedings of the Thirtieth Annual
  {ACM-SIAM} Symposium on Discrete Algorithms, {SODA} 2019, San Diego,
  California, USA, January 6-9, 2019}, pages 1844--1860. {SIAM}, 2019.

\bibitem[PRS18]{PanduranganRS18}
Gopal Pandurangan, Peter Robinson, and Michele Scquizzato.
\newblock Fast distributed algorithms for connectivity and {MST} in large
  graphs.
\newblock {\em {ACM} Trans. Parallel Comput.}, 5(1):4:1--4:22, 2018.

\bibitem[Sar06]{Sarlos06}
Tam{\'{a}}s Sarl{\'{o}}s.
\newblock Improved approximation algorithms for large matrices via random
  projections.
\newblock In {\em 47th Annual {IEEE} Symposium on Foundations of Computer
  Science {(FOCS} 2006), 21-24 October 2006, Berkeley, California, USA,
  Proceedings}, pages 143--152. {IEEE} Computer Society, 2006.

\bibitem[SW19]{SaranurakW19}
Thatchaphol Saranurak and Di~Wang.
\newblock Expander decomposition and pruning: Faster, stronger, and simpler.
\newblock In Timothy~M. Chan, editor, {\em Proceedings of the Thirtieth Annual
  {ACM-SIAM} Symposium on Discrete Algorithms, {SODA} 2019, San Diego,
  California, USA, January 6-9, 2019}, pages 2616--2635. {SIAM}, 2019.

\bibitem[Woo14]{Woodruff14}
David~P. Woodruff.
\newblock Sketching as a tool for numerical linear algebra.
\newblock {\em Found. Trends Theor. Comput. Sci.}, 10(1-2):1--157, 2014.

\bibitem[Yu21]{Yu21}
Huacheng Yu.
\newblock Tight distributed sketching lower bound for connectivity.
\newblock In D{\'{a}}niel Marx, editor, {\em Proceedings of the 2021 {ACM-SIAM}
  Symposium on Discrete Algorithms, {SODA} 2021, Virtual Conference, January 10
  - 13, 2021}, pages 1856--1873. {SIAM}, 2021.

\end{thebibliography}
